\documentclass[letterpaper,10pt]{article}

\usepackage[all,knot]{xy}
\xyoption{arc}
\usepackage{tikz}
\usetikzlibrary{decorations.markings,arrows,matrix}
\tikzset{small/.style={font=\fontsize{9}{9}\selectfont}}
\tikzset{vsmall/.style={font=\fontsize{7}{7}\selectfont}}

\usepackage{latexsym,graphics,color,epsfig}
\usepackage{rotate}
\usepackage{amsmath,amsfonts,amssymb,amsthm}
\usepackage[
    colorlinks,%
    linkcolor=blue,citecolor=red,urlcolor=blue,
]{hyperref}

\definecolor{gray}{rgb}{.5,.5,.5}

\newcommand{\xa}{\alpha}

\newcommand{\om}{\omega}
\newcommand{\Om}{\Omega}

\newcommand{\R}{{\mathbb R}}

\newcommand{\maps}{\colon}

\def\stackto #1 { \, {\stackrel{#1}{\longrightarrow}}\, }
\def\stackTo #1 { {\stackrel{#1}{\Longrightarrow}} }

\newcommand{\dt}{\underline{t}}
\newcommand{\dalpha}{\underline{\alpha}}
\newcommand{\inv}{{\rm inv}}

\renewcommand{\O}{{\rm O}}
\renewcommand{\o}{\mathfrak{o}}
\newcommand{\IO}{{\rm IO}}
\newcommand{\ISO}{{\rm ISO}}
\newcommand{\io}{\mathfrak{io}}

\newcommand{\Aut}{{\rm Aut}}
\newcommand{\aut}{\mathfrak{aut}}

\newcommand{\g}{\mathfrak{g}}
\newcommand{\h}{\mathfrak{h}}

\newcommand{\Ad}{{\rm Ad}} % adjoint rep of a group
\newcommand{\E}{\mathbf{E}}
\newcommand{\F}{\mathbf{F}}
\newcommand{\X}{\mathbf{X}}
\newcommand{\Y}{\mathbf{Y}}
\renewcommand{\P}{\mathbf{P}}

\newcommand{\G}{\mathbf{G}}
\newcommand{\Po}{\mathbf{Poinc}}
\newcommand{\Tel}{\mathbf{Tel}}  
\newcommand{\T}{\mathbf{T}}  % for the tangent 2-group \T G

\newcommand{\fake}{{\cal T}}  % The fake tangent bundle
\newcommand{\ff}{{\cal F}}  % The fake frame bundle
\newcommand{\I}{{\bf I}}  % Part of \I\ff, the extended fake frame connection
\renewcommand{\P}{\mathbf{P}}
\newcommand{\D}{\nabla} % connection on TM
\newcommand{\LC}{\widetilde\nabla} % the Levi-Civita connection
\newcommand{\fc}{\omega} % a flat Lorentz connection

\newcommand{\tfm}{{\bf 2F\!}M}  % the 2-frame 2-bundle
\newcommand{\tff}{{\bf 2}\ff}  % the fake 2-frame 2-bundle

\newcommand{\we}{\wedge}

\newtheorem{thm}{Theorem}    
\newtheorem{cor}[thm]{Corollary}
\newtheorem*{unnumbered-thm}{Theorem}
\newtheorem{defn}[thm]{Definition}
\newtheorem{example}[thm]{Example}
\newtheorem{prop}[thm]{Proposition}

\newtheorem{lemma}[thm]{Lemma}

\renewenvironment{proof}{\noindent
\textbf{Proof.}}{\hfill\rule{.6em}{.7em} \medskip}

\newcommand{\define}[1]{{\bf \boldmath{#1}}}
\newcommand{\tr}{{\mathrm{tr}}} % the trace

\newcommand{\hepth}[1]{\href{http://arxiv.org/abs/hep-th/#1}{arXiv:hep-th/#1}}

\newcommand{\physics}[1]{\href{http://arxiv.org/abs/physics/#1}{arXiv:physics/#1}}

\newcommand{\grqc}[1]{\href{http://arxiv.org/abs/gr-qc/#1}{arXiv:gr-qc/#1}}

\newcommand{\Math}[2]{\href{http://arxiv.org/abs/math.#1/#2}{arXiv:math.#1/#2}}

\newcommand{\MMath}[1]{\href{http://arxiv.org/abs/math/#1}{arXiv:math.#1}}

\newcommand{\arxiv}[1]{\href{http://arxiv.org/abs/#1/}{arXiv:#1}}

\newcommand{\webpage}[1]{{\color{blue}}\href{#1}{#1}{\color{blue}}}

\begin{document}

\title{\bf Teleparallel Gravity as a Higher Gauge Theory}

\author{\bf John C.\ Baez\\[.5em]
{\sl \small Department of Mathematics} \\[-.3em]
{\sl \small  University of California}\\[-.3em]
{\sl \small Riverside, California 92521, USA} \\
\small and \\
{\sl \small Centre for Quantum Technologies}  \\[-.3em]
{\sl \small National University of Singapore} \\[-.3em]
{\sl \small Singapore 117543}  \\
\small  \texttt{baez@math.ucr.edu} 
 \and
{\bf Derek K.\ \!Wise} \\[.5em]
{\sl \small Department of Mathematics} \\[-.3em]
{\sl \small University of Erlangen--N\"urnberg} \\[-.3em]
{\sl \small Cauerstr.~11, 91058 Erlangen, Germany} \\
\small \texttt{derek.wise@fau.de} 
}
\date{July 17, 2014}

\maketitle

\begin{abstract}
\noindent
We show that general relativity can be viewed as a higher gauge theory
involving a categorical group, or 2-group, called the teleparallel
2-group.  On any semi-Riemannian manifold $M$, we first construct a
principal 2-bundle with the Poincar\'e 2-group as its structure
2-group.  Any flat metric-preserving connection on $M$ gives a flat
2-connection on this 2-bundle, and the key ingredient of this
2-connection is the torsion.  Conversely, every flat strict
2-connection on this 2-bundle arises in this way if $M$ is simply
connected and has vanishing 2nd deRham cohomology.  Extending from the
Poincar\'e 2-group to the teleparallel 2-group, a 2-connection
includes an additional piece: a coframe field.  Taking advantage of
the teleparallel reformulation of general relativity, which uses a
coframe field, a flat connection and its torsion, this lets us rewrite 
general relativity as a theory with a 2-connection for the teleparallel 
2-group as its only field.
\end{abstract}

% -- the document --

\section{Introduction}
\label{introduction}

This paper was prompted by two puzzles in higher gauge theory.  Higher
gauge theory is the generalization of gauge theory where instead of a
connection defining parallel transport for point particles, we have a
`2-connection' defining parallel transport for particles and strings,
or an `$n$-connection' for higher $n$ defining parallel transport for
extended objects whose worldvolumes can have dimensions up to and
including $n$.  While the mathematics of this subject is increasingly
well-developed \cite{BH,BS,SSS,SW}, its potential applications to
physics remain less developed.

One puzzle concerns the Poincar\'e 2-group.  Just as ordinary gauge
theory involves choosing a Lie group, higher gauge theory involves a
Lie $n$-group. Many examples of 2-groups are known \cite{BL} and one
of the simplest is the Poincar\'e 2-group.  A 2-group can also be seen
as a `crossed module', which is a pair of groups connected by a
homomorphism, say $t \maps H \to G$, together with an action of $G$ on
$H$ obeying two equations.  For the Poincar\'e 2-group we take $G$ to
be Lorentz group, $H$ to be the translation group of Minkowski
spacetime, $t$ to be trivial, and use the usual action of Lorentz
transformations on the translation group.  The data involved here are
the same that appear in the usual construction of the Poincar\'e group
as a semidirect product.  However, the \emph{physical meaning} of the
Poincar\'e 2-group has until now remained obscure.

A `spin foam model' can be seen as a way to quantize a gauge theory or
higher gauge theory by discretizing spacetime and rewriting the path
integral as a sum \cite{B,Oriti,EPRL}.  Just as we can build spin foam
models starting from a group, we can try to do the same starting with
a 2-group.  Crane and Sheppeard proposed using the Poincar\'e 2-group
to build such a model \cite{CS}, and the mathematics needed to carry
out this proposal was developed by a number of authors \cite{BBFW,CY,Y}.  

The resulting spin foam model \cite{BW}, based on representations of the 
Poincar\'e 2-group, provides a representation-theoretic interpretation
of a model developed by Baratin and Freidel \cite{BF}.  These authors have
conjectured a fascinating relationship between this spin foam model
and Feynman diagrams in ordinary quantum field theory on Minkowski
spacetime: this spin foam model could be a `quantum model of flat
spacetime'.  However, the physical meaning of this spin foam model
remains unclear, because the corresponding classical field
theory---perhaps some sort of higher gauge theory---is not known.

This brings us to our second puzzle: the apparent shortage of
interesting classical field theories involving 2-connections.  The
reason is fairly simple.  In a gauge theory with group $G$, the most
important field is a connection, which can be seen locally as a
$\g$-valued 1-form $A$.  In a higher gauge theory based on a crossed
module $t \maps H \to G$, the most important field is a 2-connection.
This can be seen locally as a $\g$-valued 1-form $A$ together with an
$\h$-valued 2-form $B$.  However, $A$ and $B$ are not independent:
to define parallel transport along curves and surfaces in a well-behaved 
way, they must satisfy an equation, the `fake flatness condition':
\[            F = \dt(B) . \]
Here $F$ is the curvature of $A$ and we use $\dt \maps \h \to \g$, the
differential of the map $t \maps H \to G$, to convert $B$ into a
$\g$-valued 2-form. 

So far it seems difficult to get this condition to arise naturally
in field theories, except in theories without local degrees of
freedom.  For example, we can take any simple Lie group $G$, let $H$
be the vector space $\g$ viewed as an abelian Lie group, take $t$ to
be trivial, and use the adjoint action of $G$ on $H$.  This gives a
Lie 2-group called the `tangent 2-group' of $G$ \cite{BL}.  In this
case the fake flatness condition actually says that $A$ is flat: $F =
0$.  This equation is one of the field equations for 4d BF
theory, which has this Lagrangian:
\[             L = \tr(B \wedge F)  .\]
So, the solutions of 4d theory can be seen as 2-connections.  However,
in part because all flat connections are locally gauge equivalent,
there is no physical way to distinguish between two solutions of 4d BF
theory in a contractible region of spacetime.  We thus say that this
theory has no local degrees of freedom.

Another theory without local degrees of freedom, called `BFCG theory',
starts with a 2-connection for a fairly general Lie 2-group together
with some extra fields \cite{GirelliPfeifferPopescu, MartinsMikovic}.
The equations of motion again imply fake flatness.  A 
proposal due to Mikovi\'c and Vojinovi\'c \cite{MikovicVojinovic}
involves modifying the BFCG action for the Poincare 2-group to obtain
a theory equivalent to general relativity.  However, after this
modification the equations of motion no longer imply fake flatness.
Moreover, for solutions, the $\h$-valued 2-form in the would-be
2-connection is always zero.  So, the problems of finding a clear
geometrical interpretation of Poincar\'e 2-connections, and finding a
physically interesting theory involving such 2-connections, still
stand.

In this paper we suggest a way to solve both these problems in one
blow.  The idea is to treat gravity in 4d spacetime as a higher gauge
theory based on the Poincar\'e 2-group, or some larger 2-group.  We do
this using a reformulation of general relativity called `teleparallel
gravity'.  This theory is locally equivalent to general relativity, at
least in the presence of spinless matter.  Einstein studied it
intensively from 1928 to 1931 \cite{Sauer}, and he also had a
significant correspondence on the subject with \'Elie Cartan
\cite{EC}.

As in general relativity, the idea in teleparallel gravity is to start
with a metric on the spacetime manifold $M$ and then choose a
metric-compatible connection $\omega$ on the tangent bundle of $M$.
However, instead of taking $\omega$ to be the torsion-free (and
typically curved) Levi-Civita connection, we take $\om$ to be flat
(and typically with nonzero torsion).  There is always a way to do
this, at least locally.  To do it, we can pick any orthonormal coframe
field $e$ and let $\om$ be the unique flat metric-compatible
connection for which the covariant derivative of $e$ vanishes.  That
is, we define $\om$ by declaring a vector $v$ to be parallel
transported if $e(v)$ is constant along the path.  This choice of
$\om$ is called the `Weitzenb\"ock connection' for $e$.  And then,
remarkably, one can convert all of the standard equations in general
relativity into equations involving the coframe and its Weitzenb\"ock
connection, with no reference to the Levi-Civita connection.

How is this related to the Poincar\'e 2-group?  Recall that in the
Poincar\'e 2-group, $G$ is the Lorentz group and $H$ is the group of
translations of Minkowski spacetime.  The Weitzenb\"ock connection $\om$
can locally be seen as a $\g$-valued 1-form, and its torsion, defined
by
\[
     T(v,w) = \D_v w -\D_w v -[v,w]
\]
can locally be seen as an $\h$-valued 2-form, namely the exterior
covariant derivative $d_\om e$.  So, $\om$ and $T$ are the right
sort of objects---at least locally, but in fact globally---to form a
2-connection with the Poincar\'e 2-group as gauge 2-group.  Even
better, since the map $t \maps H \to G$ is trivial, the fake flatness
condition says the Weizenb\"ock connection $\om$ is flat,
\emph{which is true}.

In short, general relativity can be reformulated as a theory involving
a flat metric-compatible connection $\om$.  Even though the connection
is flat, this theory has local degrees of freedom because the torsion
$T$ is nonzero and contains observable information about the local geometry
of spacetime.  Furthermore, the pair $(\om,T)$ fit together to form a
2-connection, and the relevant gauge 2-group is the Poincar\'e 2-group.

However, teleparallel gravity involves not just the Weitzenb\"ock
connection, but also the coframe field $e$.  So, for a higher gauge
theory interpretation of teleparallel gravity, we would also like to
understand the coframe field as part of a 2-connection.  For this, it
is helpful to recall that an ordinary Poincar\'e connection, say
$A$, consists of two parts: a Lorentz connection $\omega$ and a
1-form $e$ valued in the Lie algebra of the translation group.  When
$e$ obeys a certain nondegeneracy condition, it is precisely the same
as a coframe field.  The curvature of $A$ then consists of two parts,
the curvature of $\om$:
\[              R = d \om + \om \wedge \om \]
and the torsion:
\[              T = d_\omega e.  \]

To take advantage of these well-known facts, we enlarge the Poincar\'e
2-group to the `teleparallel 2-group'.  This comes from the crossed
module $t \maps H \to G$ where $G$ is the Poincar\'e group, $H$ is the
translation group of Minkowski spacetime, $t$ is the inclusion, and
$G$ acts on $H$ by conjugation, using the fact that $H$ is a normal
subgroup of $G$.  A 2-connection with this gauge 2-group turns out to
consist of three parts:
\begin{itemize}
\item a flat Lorentz group connection $\omega$,
\item a 1-form $e$ valued in the Lie algebra of the translation group,
\item the torsion $T = d_\omega e$.
\end{itemize}
These are precisely the data involved in teleparallel gravity.

In what follows we flesh out this story.  We begin in Section
\ref{higher} with a review of higher gauge theory, including Lie
groupoids, Lie 2-groups, 2-bundles and 2-connections.  To minimize
subtleties that are irrelevant here, and make the paper completely
self-contained, we work in the `strict' rather than the fully general
`weak' framework.  In Section \ref{poincare}, we begin by describing the
relation between teleparallel geometry and Poincar\'e 2-connections.
The highlight of this section is Theorem \ref{thm:main}, which among
other things describes Poincar\'e 2-connections on a principal
2-bundle called the `2-frame 2-bundle' canonically associated to any
semi-Riemannian manifold.  

In Section \ref{teleparallel-2-group}, we introduce the teleparallel
2-group.  Theorem \ref{telpq-conn} does for this 2-group what 
\ref{thm:main} did for the Poincar\'e 2-group.  Then, we describe how
to express the Lagrangian for teleparallel gravity as a function of a
teleparallel 2-connection. This action is not invariant under all
teleparallel 2-group gauge transformations, but only under those in a
sub-2-group.  This may seem disappointing at first, but it mirrors
what we are already familiar with in the Palatini formulation of
general relativity, where the fields can be seen as forming a
Poincar\'e connection, but the action is only invariant under gauge
transformations lying in the Lorentz group.  This phenomenon in
Palatini gravity can be neatly understood using Cartan geometry
\cite{Wise2}.  So, it is natural to expect that the similar phenomemon
in teleparallel gravity can be understood using `Cartan 2-geometry',
and we present some evidence that this is the case.  We also consider
what happens when we go beyond the `strict' framework discussed here
to the more general `weak' framework.

\section{Higher gauge theory}
\label{higher} 

Here we introduce all the higher gauge theory that we will need in
this paper.  First we explain the Poincar\'e 2-group $\Po(p,q)$ and
its Lie 2-algebra.  Then we explain principal 2-bundles and show that
any semi-Riemannian manifold of signature $(p,q)$ has a principal
2-bundle over it whose structure 2-group is $\Po(p,q)$.  We take a
businesslike approach, often sacrificing generality and elegance for
efficiency.  For a more well-rounded introduction to higher gauge
theory, see our review article \cite{BH} and the more advanced
references therein.

\subsection{Lie groupoids}
\label{Lie groupoids}

The first step towards higher gauge theory is to generalize the
concept of `manifold' to a kind of space that has, besides 
points, also arrows between points.  There are many ways to do 
this, all closely related but differing in technical details.  
Here we use the concept of a `Lie groupoid'.  A \define{groupoid} 
is a category where every morphism has an inverse; we can 
visualize the objects of a groupoid as points and the morphisms 
as arrows.  A Lie groupoid is basically a groupoid where the set 
of points forms a smooth manifold, and so does the set of arrows.  
More precisely:

\begin{defn} 
A \define{Lie groupoid} $\X$ is a groupoid where:
\begin{itemize}
\item the collection of objects is a manifold, say $\X_0$;
\item the collection of morphisms is a manifold, say $\X_1$;
\item the maps $s,t \maps \X_1 \to \X_0$ sending each
morphism to its source and target are smooth;
\item the map sending each object to its identity morphism is smooth;
\item the map sending each morphism to its inverse is smooth;
\item the set of composable pairs of morphisms is a submanifold of 
$\X_1 \times \X_1$, and composition is a smooth map from this
submanifold to $\X_1$.  
\end{itemize}
\end{defn}
To ensure that the set of composable pairs of morphisms is a
submanifold, it suffices to assume that the map $s$, or 
equivalently $t$, is a submersion.  This assumption is commonly 
taken as part of the definition of a Lie groupoid, and the reader 
is welcome to include it, since it holds in all our examples, but 
we will not actually need it.

We will use the obvious naive notion of map between Lie groupoids:
\begin{defn} 
Given Lie groupoids $\X$ and $\Y$, a \define{map} $f \maps \X \to
\Y$ is a functor for which:
\begin{itemize}
\item the map sending objects to objects, say $f_0 \maps \X_0 \to 
\Y_0$, is smooth;
\item the map sending morphisms to morphisms, say $f_1 \maps \X_1 
\to \Y_1$, is smooth.
\end{itemize}
\end{defn}
There is another more general notion of map between Lie groupoids, and
in this more general context Lie groupoids may be identified with
`differentiable stacks' \cite{BX,Heinloth}.  However, in this paper we
only need maps of the above type.

Many interesting examples of Lie groupoids can be found in Mackenzie's
book \cite{Mackenzie}.  Here are two rather trivial examples we will
need:

\begin{example} 
\label{2space.1}
Any manifold can be seen as a Lie groupoid with only identity
morphisms. In what follows, we freely treat manifolds as Lie groupoids
in this way. Note that a map between Lie groupoids of this type is the
same as a smooth map between manifolds.
\end{example}

\begin{example}
\label{2space.4}
Given Lie groupoids $\X$ and $\Y$, there is a \define{product} Lie
groupoid $\X \times \Y$ with $(\X \times \Y)_0 = \X_0 \times \Y_0$ and
$(\X \times \Y)_1 = \X_1 \times \Y_1$, where the source and target
maps, identity-assigning map and composition of morphisms are defined
componentwise.  Note that this product comes with \define{projection}
maps from $\X \times \Y$ to $\X$ and to $\Y$.
\end{example}

\subsection{2-Groups}
\label{2-groups}

In general, a Lie 2-group is a Lie groupoid equipped with a
multiplication that obeys the group axioms {\it up to isomorphism}
\cite{BL, Henriques, Schommer-Pries}.  But the 2-groups needed
in this paper, including the Poincar\'e 2-group, are all 
`strict': the group laws hold on the nose, as equations.  So, we 
need only the definition of strict Lie 2-groups:

\begin{defn} A \define{(strict) Lie 2-group} $\G$ is a Lie groupoid 
equipped with an \define{identity} object $1 \in \G$, a
\define{multiplication} map $m \maps \G \times \G \to \G$, and an
\define{inverse} map $\inv \maps \G \to \G$ such that the usual group
axioms hold.
\end{defn}

\noindent
Note that given a Lie 2-group $\G$, the manifold of objects 
$\G_0$ forms a Lie group with multiplication $m_0$, and the 
manifold of morphisms $\G_1$ forms a Lie group with 
multiplication $m_1$. So, a 2-group consists of two groups, but 
with further structure given by the source and target maps $s,t 
\maps \G_1 \to \G_0$, the identity-assigning map $i \maps \G_0 
\to \G_1$, and composition of morphisms.

We are mainly interested in physics on spacetimes with several 
space dimensions and one time dimension, but the construction of 
the Poincar\'e 2-group works for any signature.  So, let us fix 
natural numbers $p,q \ge 0$ and define $\R^{p,q}$ to be the 
vector space $\R^{p+q}$ equipped with the metric
\[  ds^2 = dx_1^2 + \cdots + dx_p^2 - dx_{p+1}^2 - \cdots - dx_{p+q}^2 .\]
We often think of  $\R^{p,q}$ as a group, and call it the 
\define{translation group}.  We write $\O(p,q)$ to mean the group 
of linear isometries of $\R^{p,q}$.  With some abuse of language, 
let us call this group the \define{Lorentz group}.  Similarly, we 
define the \define{Poincar\'e group} to be the semidirect product
\[            \IO(p,q) = \O(p,q) \ltimes \R^{p,q}  \]
where $\O(p,q)$ acts on $\R^{p,q}$ as linear transformations in 
the obvious way.  Multiplication in the Poincar\'e group is given 
by 
\[
        (g,h)(g',h') = (gg',h +gh'),
\]
where an element $(g,h)$ consists of an element $g \in \O(p,q)$ and an
element $h \in \R^{p,q}$.  We denote the Poincar\'e group as
$\IO(p,q)$ because it is also called the \define{inhomogeneous
orthogonal group}.  We will not need the smaller inhomogenous
special orthogonal group $\ISO(p,q)$.

The Poincar\'e 2-group is very similar to the Poincar\'e group.  The
ingredients used to build it are just the same: the Lorentz group and
its action on the translation group.  The difference is that now the
translations enter at a higher categorical level than the Lorentz
transformations: namely, as morphisms rather than objects.

\begin{defn}  The \define{Poincar\'e 2-group} $\Po(p,q)$ is
the Lie 2-group where:
\begin{itemize}
\item The Lie group of objects is the Lorentz group $\O(p,q)$.
\item The Lie group of morphisms is the Poincar\'e group $\IO(p,q)$.
\item The source and target of a morphism $(g,h) \in \IO(p,q)$ 
are both equal to $g$.
\item The composite morphism $(g,h') \circ (g,h)$ is 
$(g,h' + h)$, where addition is done as usual in $\R^{p,q}$.
\end{itemize}
\end{defn}

In what follows, we will also need to see the Lorentz group as
a 2-group of a degenerate sort.  This works as follows:

\begin{example} 
\label{2group.1}
In Example \ref{2space.1} we saw that any manifold may be seen as a
Lie groupoid with only identity morphisms.  As a corollary, any Lie
group $G$ may be seen as a 2-group whose morphisms are all identity
morphisms.  This 2-group has $G$ as its group of objects and also $G$
as its group of morphisms; by abuse of language we call this 2-group
simply $G$.
\end{example}

\subsection{Actions of Lie 2-groups}
\label{Actions}

Just as Lie groups can act on manifolds, Lie 2-groups can act
on Lie groupoids.  Here we only consider `strict' actions:

\begin{defn} Given a Lie 2-group $\G$, a \define{(strict) left $\G$ 2-space}
is a Lie groupoid $\X$ equipped with a map $\alpha \maps 
\G \times \X \to \X$ obeying the usual axioms
for a left group action.  Similarly, a \define{(strict) right $\G$ 2-space}
is a Lie groupoid $\X$ equipped with a map $\alpha \maps 
\X \times \G \to \X$ obeying the usual axioms for a right group action.
\end{defn}

\begin{defn} Given a Lie 2-group $\G$ and left 
$\G$ 2-spaces $\X$ and $\Y$, we define a \define{(strict) map of
$\G$ 2-spaces} to be a map of Lie groupoids $f \maps \X \to \Y$ such
that acting by $\G$ and then mapping by $f$ is the same as mapping and
then acting.  In other words, this diagram commutes:
\[
\xymatrix{
\G \times \X \ar[r]^{\alpha_\X} \ar[d]_{1 \times f} & \X \ar[d]^f \\
\G \times \Y \ar[r]_{\alpha_Y} & \Y \\
}
\]
where $\alpha_\X$ is the action of $\G$ on $\X$, and 
$\alpha_\Y$ is the action of $\G$ on $\Y$.  
\end{defn}

The examples we need are these:

\begin{example}
\label{G-2-space.1}
Using the idea in Examples \ref{2space.1} and \ref{2group.1}, 
any manifold can be seen as a Lie groupoid, and any Lie group
can be seen as a Lie 2-group.   Continuing this line of thought,
if the manifold $X$ is a left $G$ space, we can think of
it as a left 2-space for the 2-group.  Similarly, any right $G$ space
can be seen as a right $G$ 2-space for this 2-group.
\end{example}

\begin{example}
\label{G-2-space.2}
Given a Lie 2-group $\G$, the multiplication $m \maps \G \times \G\to \G$
makes $\G$ into both a left $\G$ 2-space and a right $\G$ 2-space.
\end{example}

\begin{example}
\label{G-2-space.3}
If we treat the Lorentz group $\O(p,q)$ as a Lie 2-group following the
ideas in Example \ref{2group.1}, the result is a `sub-2-group' of
$\Po(p,q)$.  This Lie 2-group has a left action on
$\Po(p,q)$, coming from left multiplication.  Explicitly, the action
\[    \alpha \maps \O(p,q) \times \Po(p,q) \to \Po(p,q)  \]
has
\[   \alpha_0 \maps \O(p,q) \times \O(p,q) \to  \O(p,q) \]
given by multiplication in the Lorentz group, and
\[   \alpha_1 \maps \O(p,q) \times \IO(p,q) \to \IO(p,q) \]
given by restricting multiplication in the Poincar\'e group.
\end{example}

\subsection{Crossed modules}
\label{crossed modules}

A `Lie crossed module' is an alternative way of
describing a Lie 2-group.  The reader can find the full 
definition of a Lie crossed module elsewhere \cite{BH,BL}; we do
not need it here.  All we need to know is that a Lie crossed 
module is a quadruple $(G,H,t,\alpha)$ that we can extract from
a Lie 2-group $\G$ as follows:
\begin{itemize}
\item $G$ is the Lie group $\G_0$ of objects of $\G$;
\item $H$ is the normal subgroup of $\G_1$ consisting
of morphisms with source equal to $1 \in G$;
\item $t \maps H \to G$ is the homomorphism sending each 
morphism in $H$ to its target;
\item $\alpha$ is the action of $G$ as automorphisms of $H$ 
defined using conjugation in $\G_1$ as follows: 
$\alpha(g) h = 1_g h {1_g}^{-1}$, 
where $1_g\in \G_1$ is the identity morphism of $g\in G$. 
\end{itemize}
\noindent Conversely, any Lie crossed module gives a Lie 2-group.

If we take all the data in a Lie crossed module and differentiate it, 
we get:
\begin{itemize}
\item
the Lie algebra $\g$ of $G$,
\item
the Lie algebra $\h$ of $H$,
\item 
the Lie algebra homomorphism $\dt \maps \h \to \g$ obtained by
differentiating $t \maps H \to G$, and
\item 
the Lie algebra homomorphism $\dalpha \maps \g \to \aut(H)$ obtained
by differentiating $\alpha \maps G \to \Aut(H)$.
\end{itemize}
Here we write $\dt$ and $\dalpha$ instead of $dt$ and $d\alpha$ 
to reduce confusion in later formulas involving these maps and 
also differential forms, where $d$ stands for the exterior 
derivative. The quadruple $(\g, \h, \dt, \dalpha)$ is called an 
\define{infinitesimal crossed module} \cite{BC,BS}.  It will be 
important in describing 2-connections on a principal 2-bundle 
with structure 2-group $\G$.

The crossed module $(G,H,t,\alpha)$ coming from the Poincar\'e 
2-group works as follows:
\begin{itemize}
\item $G$ is the Lorentz group $\O(p,q)$;
\item $H$ is the translation group $\R^{p,q}$ viewed as an abelian Lie group;
\item $t$ is trivial.
\item $\alpha$ is the obvious representation of $\O(p,q)$ on $\R^{p,q}$.
\end{itemize}
Differentiating all this, we obtain an infinitesimal crossed
module where:
\begin{itemize}
\item $\g = \o(p,q)$;
\item $\h = \R^{p,q}$ viewed as an abelian Lie algebra;
\item $\dt$ is trivial;
\item $\dalpha$ is the obvious representation of $\o(p,q)$ on $\R^{p,q}$.
\end{itemize}

\subsection{Principal 2-bundles}
\label{principal 2-bundles}

In what follows we take a lowbrow, pragmatic approach to 2-bundles. In
particular, we only define `strict' 2-bundles, since those are all we
need here.  The interested reader can find more sophisticated
treatments elsewhere \cite{Bartels:2004, Breen:2008, BreenMessing,
SW}.

Just as a bundle involves a `projection' $p \maps E \to M$ that is a
map between manifolds, a strict 2-bundle involves a map $p \maps \E
\to M$.  Here $\E$ is a Lie groupoid, but for our applications $M$ is
a mere manifold, regarded as a Lie groupoid as in Example
\ref{2space.1}.

The key property we require of a bundle is `local triviality'.  
For this, note that given any open subset $U \subseteq M$, there 
is a Lie groupoid $\E|_U$ whose objects and morphisms are 
precisely those of $\E$ that map under $p$ to objects and 
morphisms lying in $U$.  Then we can restrict $p$ to $\E|_U$ and 
obtain a map we call $p|_U \maps \E|_U \to U$.  The local 
triviality assumption says that every point $x \in M$ has a 
neighborhood $U$ such that $p|_U \maps \E|_U \to U$ is
a `trivial 2-bundle'.  More precisely:

\begin{defn}
\label{2-bundle}
  A \define{(strict) 2-bundle} consists of:
  \begin{itemize}
    \item
      a Lie groupoid $\E$ (the \define{total 2-space}),
    \item
      a manifold $M$ (the \define{base space}),
    \item
      a map $p \maps \E \to M$ 
      (the \define{projection}), and
    \item a Lie groupoid $\F$ (the \define{fiber}),
  \end{itemize}
such that for every point $p \in M$ there is an open neighborhood $U$ 
containing $p$ and an isomorphism
  \[
   t \maps E|_U \to U \times \F,
  \]
called a \define{local trivialization}, such that this diagram commutes:
\[ 
\xymatrix @!0
{ E|_U
 \ar [ddr]_{p}
  \ar[rr]^{t} 
  & &  
  U \times \F
  \ar[ddl]^{p_{{}_U}}
  \\  \\ &
   U 
 }
\]
\end{defn}

\begin{defn}
Given a Lie 2-group $\G$, a \define{(strict) principal $\G$ 2-bundle}
is a 2-bundle $p \maps \P \to M$ where:
\begin{itemize}
\item the fiber $\F$ is $\G$,
\item the total 2-space $\P$ is a right $\G$ 2-space, and
\item for each point $p \in M$ there is a neighborhood $U$ containing
$p$ and a local trivialization $t \maps E|_U \to U \times \F$
that is an isomorphism of right $\G$ 2-spaces.
\end{itemize}
\end{defn}
It is worth noting that just as a 2-group consists of two groups 
and some maps relating them, a principal 2-bundle consists of two 
principal bundles and some maps relating them.  Suppose $p \maps 
\P \to M$ is a principal $\G$ 2-bundle.  Then the bundle of 
objects, $p_0 \maps \P_0 \to M$, is a principal $\G_0$ bundle, 
and the bundle of morphisms, $p_1 \maps \P_1 \to M$, is a 
principal $\G_1$ bundle.  These are related by the source
and target maps $s,t \maps \P_1 \to \P_0$, the identity-assigning 
map $i \maps \P_0 \to \P_1$, and the map describing composition 
of morphisms.  All these are maps between bundles over $M$.
 
We can build principal 2-bundles using transition functions.
First recall the situation in ordinary gauge theory: suppose $G$ 
is a Lie group and $M$ a manifold.  In this case we can build a 
$G$ bundle over $M$ using transition functions.  To do this, we 
write $M$ as the union of open sets or \define{patches} $U_i 
\subseteq M$:
\[ M = \bigcup_i U_i. \]
Then, choose a smooth transition
function on each double intersection of patches:
\[ g_{ij} \maps U_i \cap U_j \to G .\]
These transition functions give gauge transformations.  We can 
build a principal $G$ bundle over all of $M$ by gluing together 
trivial bundles over the patches with the help of these gauge 
transformations.  However, this procedure will only succeed if 
the transition functions satisfy a consistency condition on each 
triple intersection:
\[ g_{ij}(x) g_{jk}(x) = g_{ik}(x) \]
for all $x \in U_i \cap U_j \cap U_k$.  This equation is called 
the \define{cocycle condition}.  

A similar recipe works for higher gauge theory. Now let $\G$ be a Lie
2-group with $G$ as its Lie group of objects.  To build a principal
$\G$ 2-bundle, it suffices to choose transition functions on double
intersections of patches:
\[ g_{ij} \maps U_i \cap U_j \to G \]
such that the cocycle condition holds.  Conversely, given a 
strict principal $\G$ 2-bundle, a choice of open neighborhoods 
containing each point and local trivializations gives rise to 
transition functions obeying the cocycle condition.  In a more 
general `weak' principal 2-bundle, the cocycle condition would 
need to hold only up to isomorphism \cite{Bartels:2004}.  While 
this generalization is very important, none of our work here 
requires it.

For now we only give a rather degenerate example of a principal
2-bundle.  The next section will introduce a more interesting 
example, one that really matters to us: the `2-frame 2-bundle'.

\begin{example} 
\label{2bundle.1}
Suppose $p \maps P \to M$ is a principal $G$ bundle.  Then we can
regard $G$ as a Lie 2-group as in Example \ref {2group.1}, and $P$ as
a right 2-space of this Lie 2-group as in Example \ref{G-2-space.1}.
Then $p \maps P \to M$ becomes a principal 2-bundle with this Lie
2-group as its structure 2-group.
\end{example}

\subsection{Associated 2-bundles}
\label{2-frame 2-bundle}

Just as principal bundles have associated bundles, principal 
2-bundles have associated 2-bundles.  Given a principal $\G$ 
2-bundle $p \maps \P \to M$ and a left $\G$ 2-space $\F$, we can 
construct an \define{associated 2-bundle}
\[       q \maps \P \times_G \F \to M  \]
with fiber $\F$.  To do this, we first construct a Lie groupoid 
$\P \times_\G \F$.   This Lie groupoid has a manifold of objects
$(\P \times_\G \F)_0$ and a manifold of morphisms $(\P \times_\G 
\F)_1$.   Both these are given as quotient spaces:
\[   (\P \times_\G \F)_i = \frac{\P_i \times \F_i}{(xg, f) \sim (x, gf)} \]
where $xg \in \P_i$ is the result of letting $g \in \G_i$ act on the
right on $x \in \P_i$, and $gf \in \F_i$ is the result of letting $g
\in \G_i$ act on the left on $f \in \F_i$.  One can check that these
quotient spaces are indeed manifolds.  One can also check that the
usual composition of morphisms in $\P \times \F$ descends to $\P
\times_\G \F$, and that $\P \times_\G \F$ is a Lie groupoid.  Then, we
can define a map
\[       q \maps \P \times_G \F \to M  \]
sending the equivalence class of any object $(x,f)$ to the object 
$p(x)$, and doing the only possible thing on morphisms.  Finally, 
we can check that $q \maps \P \times_G \F \to M$ is a 2-bundle 
with fiber $\F$.

To connect teleparallel gravity to higher gauge theory, we 
take advantage of the fact that any manifold with a metric 
comes equipped with principal 2-bundle whose structure 2-group is 
the Poincar\'e 2-group.  We call this the `2-frame 2-bundle'. 
To build it, we start with a manifold $M$ equipped 
with a semi-Riemannian metric $g$ of signature $(p,q)$.  Let 
$p \maps FM \to M$ be the bundle of orthonormal frames, a principal
$\O(p,q)$ bundle.  As in Example \ref{2bundle.1} we can think of
$\O(p,q)$ as a Lie 2-group and
$p \maps FM \to M$ as a principal $\O(p,q)$ 2-bundle.

Next, recall from Example \ref{G-2-space.3} that the Poincar\'e
2-group $\Po(p,q)$ is a left 2-space for the 2-group $\O(p,q)$.  
This lets us form the associated 2-bundle $q \maps FM \times_{\O(p,q)}
\Po(p,q) \to M$.  

\begin{defn} 
\label{2FM} 
Given a manifold $M$ equipped with a metric of signature
$(p,q)$, we set
\[      \tfm = FM \times_{\O(p,q)} \Po(p,q) \]
and define the \define{2-frame 2-bundle} of $M$ to be 
\[           q \maps \tfm \to M . \]
\end{defn}

Note that since $\Po(p,q)$ is a right 2-space for itself, 
$\tfm$ is also a right $\Po(p,q)$ 2-space.  Using this fact,
it is easy to check the following:

\begin{prop}  \label{2-frame}  Given a manifold $M$ equipped with
a metric of signature $(p,q)$, the 2-frame 2-bundle $\tfm$ is a
principal 2-bundle with structure 2-group $\Po(p,q)$. 
\end{prop}

\subsection{2-Connections}
\label{2-connections}

Just as a connection on a bundle over a manifold $M$ allows us to
describe parallel transport along curves in $M$, a `2-connection' on a
2-bundle over $M$ allows us to describe parallel transport along
\emph{curves and surfaces}.  Just as a connection on a $G$ bundle can
be described locally as 1-form taking values in the Lie algebra $\g$,
a connection on a trivial $\G$ 2-bundle can be described locally as a
$\g$-valued 1-form $A$ together with an $\h$-valued 2-form.  
in order to consistently define parallel transport along curves and
surfaces, $A$ and $B$ need to be related by an equation called 
the `fake flatness condition'.  The 
concept of `fake curvature' was first introduced by Breen and Messing
\cite{BreenMessing}, but only later did it become clear that a consistent 
theory of parallel 
transport in higher gauge theory requires the fake curvature to 
vanish \cite{BS,SW}.  

Suppose $\G$ is a 2-group whose infinitesimal crossed module is
$(\g,\h,\dt,\dalpha)$.  Then a \define{2-connection} on the trivial
principal $\G$ 2-bundle over a smooth manifold $M$ consists of
\begin{itemize}
\item
a $\g$-valued 1-form $A$ on $M$ and
\item
an $\h$-valued 2-form $B$ on $M$ 
\end{itemize}
such that the \define{fake flatness} condition holds:
\[            dA + A \wedge A = \dt(B) . \]
Here we use $\dt \maps \h \to \g$, to convert $B$ into a
$\g$-valued 2-form.  As usual, $A \wedge A$ really stands for
$\frac{1}{2}[A,A]$, defined using the wedge product of 1-forms together
with the Lie bracket in $\g$.

To describe 2-connections on a general strict $\G$ 2-bundle, we need to 
know how 2-connections on a trivial 2-bundle transform under 
gauge transformations.  Here we only consider `strict' gauge transformations:

\begin{defn}
\label{gauge transformation}
Given a principal $\G$ 2-bundle $p \maps \P \to M$, a \define{(strict)
gauge transformation} is a map of right $\G$ spaces $f \maps \P \to
\P$ such that this diagram commutes:
\[ 
\xymatrix @!0
{ \P
 \ar [ddr]_{p}
  \ar[rr]^{f} 
  & &  
  \P
  \ar[ddl]^{p}
  \\  \\ &
   M 
 }
\]
\end{defn}

It is easy to check that strict gauge transformations on the trivial
principal $\G$ 2-bundle over $M$ are in one-to-one correspondence with
smooth functions $g \maps M \to G$, where $G$ is the group of objects
of $\G$.   The proof is exactly like the proof for ordinary principal
bundles.  So, henceforth we identify gauge transformations on the trivial
$\G$ 2-bundle over $M$ with functions of this sort.

Now suppose we have a 2-connection $(A,B)$ on the trivial principal
$\G$ 2-bundle over $M$.  By definition, a gauge transformation $g
\maps M \to G$ acts on this 2-connection to give a new 2-connection 
$(A',B')$ as follows:
\[       
\begin{array}{ccl}
A' &=& g A g^{-1} \; + \; g\, dg^{-1}   \\[.3em]
B' &=& \alpha(g)(B) 
\end{array}
\]
The second formula deserves a bit of explanation.  Here we are
composing $\alpha \maps G \to \Aut(H)$ with $g \maps M \to
G$ and obtaining an $\Aut(H)$-valued function $\alpha(g)$, which 
then acts on the $\h$-valued 2-form $B$ to give a new $\h$-valued 
2-form $\alpha(g)(B)$.  

Now, suppose we have a principal $\G$ 2-bundle $p \maps P \to M$ built
using transition functions 
\[    g_{ij} \maps U_i \cap U_i \to G \]
as described in Section \ref{principal 2-bundles}.  Note that each
function $g_{ij}$ can be seen as a gauge transformation.  To equip $P$
with a \define{strict 2-connection}, we first put a strict
2-connection on the trivial 2-bundle over each patch $U_i$.  So, on
each patch we choose a $\g$-valued 1-form $A_i$ and an $\h$-valued
2-form $B_i$ obeying
\[            dA_i + A_i \wedge A_i = \dt(B_i)  \]
Then, we require that on each intersection of patches $U_i \cap U_j$,
the 2-connection $(A_i,B_i)$ is the result of applying the gauge
transformation $g_{ij}$ to $(A_j,B_j)$:
\begin{eqnarray}
A_i &=& g_{ij} A_j g_{ij}^{-1} \, + \, g_{ij}\, dg_{ij}^{-1} 
\label{local-2conn-1} \\
B_i &=& \alpha(g_{ij})(B_j) 
\label{local-2conn-2} 
\end{eqnarray}
For more details, including the more general case of weak
2-connections on weak 2-bundles, see \cite{BH} and the references
therein.

We can also define strict 2-connections in terms of {\em global} data:
\begin{prop}
\label{global-2conn}
Suppose $p \maps \P\to B$ is a principal $\G$ 2-bundle.  Then a
strict 2-connection $(A,B)$ on $\P$ is the same as a connection $A$ on
the principal $\G_0$ bundle $p_0 \maps \P_0\to B$ together with a
$(\P_0 \times_{\G_0} \h)$-valued 2-form $B$ obeying the fake flatness
condition.
\end{prop}

\begin{proof}
The equivalence follows simply by inspection of the local data used to
build a strict 2-bundle with strict 2-connection.  Choosing a suitable
cover $U_i$ of $B$, the principal $\G$ 2-bundle $\P$ and the principal
$\G_0$ bundle $\P_0$ are constructed via the same transition functions
$ g_{ij} \maps U_i \cap U_j \to \G_0 $.  Then (\ref{local-2conn-1}) is
precisely the condition that assembles a family of Lie algebra-valued
1-forms into a connection on $\P_0$.  Similarly, (\ref{local-2conn-2})
is just the equation relating the local descriptions of a 2-form with
values in the associated bundle $\P_0\times_{\G_0} \h$.  The only
additional condition for the 2-connection $(A,B)$ is fake flatness, so
imposing this, we are done.
\end{proof} 

This theorem gives us a way to understand the fake flatness condition
in a global way.  The curvature of a connection $A$ on a principal
$G$ bundle $P$ is a $(P\times_G \g)$-valued 2-form $F$, where
the vector bundle $P\times_G \g$ is built using the adjoint representation 
of $G$.  The map $\dt \maps \h \to \g$ in the infinitesimal crossed module 
is an intertwiner of representations of $G$:
\[
            \Ad(g)\dt(X) = \dt(\xa'(g)X) \qquad \forall g\in G, X\in \h
\] 
where $\xa'$ denotes the representation of $G$ on $\h$ given by
differentiating $\xa$.  Thus $\dt$ induces a vector bundle map
$\dt\maps P\times_G \h \to P\times_G \g$.  Using this, the fake
flatness condition
\[
         F = \dt(B) 
\]  
makes sense globally as an equation between $(P\times_G \g)$-valued forms. 

\subsection{Curvature}
\label{curvature}

There are three kinds of curvature for 2-connections.  Suppose $(A,B)$
is a strict 2-connection on a $\G$ 2-bundle $P \to M$.  By working
locally, we may choose a trivialization for $P$ and treat $A$ as a
$\g$-valued 1-form and $B$ as an $\h$-valued 2-form.  This simplifies
the discussion a bit.

First, just as in ordinary gauge theory, we may define the 
\define{curvature} to be the $\g$-valued 2-form given by:
\[                 F = dA + A \wedge A  .\]
Second, we define the \define{fake curvature} to be the 
$\g$-valued 2-form $F - \dt(B)$.  However, this must equal zero.
Third, we define the \define{2-curvature} to be the $\h$-valued 
3-form given by:
\[                 G = dB + \dalpha(A) \wedge B  .\]
Beware: the symbol $G$ here has nothing to do with the
group $G$.  In the second term on the right-hand side, we compose
$\dalpha \maps \g \to \aut(H)$ with the $\g$-valued 1-form $A$ 
and obtain an $\aut(H)$-valued function $\dalpha(g)$.  Then we 
wedge this with $B$, letting $\aut(H)$ act on $\h$ as part of 
this process, and obtain an $\h$-valued 2-form.

The intuitive idea of 2-curvature is this: just as the curvature 
describes the holonomy of a connection around an infinitesimal 
loop, the 2-curvature describes the holonomy of a 2-connection 
over an infinitesimal 2-sphere.  This can be made precise using 
formulas for holonomies over surfaces \cite{MartinsPicken:2007, 
SW}, which we will not need here.

If the 2-curvature of a 2-connection vanishes, the holonomy over 
a surface will not change if we apply a smooth homotopy to that 
surface while keeping its edges fixed.  A 2-connection whose 
curvature and 2-curvature both vanish truly deserves to be called 
\define{flat}.  

\section{Teleparallel geometry and the Poincar\'e 2-group}
\label{poincare}

With tools of higher gauge theory in hand, we turn now to our main
geometric and physical applications---teleparallel geometry and
teleparallel gravity.

The simplest version of teleparallel gravity may be viewed as a
rewriting of Einstein gravity in which torsion, as opposed to
curvature, plays the lead role.  This results in a theory that is
conceptually quite different from general relativity: if gravity is
interpreted within the teleparallel framework, then Einstein's
original vision appears wrong on several counts.  Unlike in general
relativity, the spacetime of teleparallel gravity is {\em flat}.  As a
consequence, it is possible to compare vectors at distant points to
decide, for example, whether the velocity vectors of two distant
observers are parallel (hence the term `teleparallelism' or
`distant parallelism').  Flat spacetime clearly flies in the face of
Einstein's geometric picture of gravity as spacetime curvature: in
fact, in teleparallel theories, gravity is a {\em force}.

With these features, teleparallel gravity sounds like such a throwback
to the Newtonian understanding of gravity that it would be easy to
dismiss, except for one fact: teleparallel gravity is {\em locally
equivalent} to general relativity.  

Our main interest in teleparallel gravity here is that it involves a
flat connection $\om$ and its torsion $d_\om e$.  As we shall see,
these are precisely the data needed for a Poincar\'e 2-connection on
the 2-frame 2-bundle of spacetime---or more precisely, a 2-bundle
isomorphic to this.  In what follows we briefly sketch some of the
main ideas of teleparallel gravity, and discuss how it may be viewed
as a higher gauge theory.  For more on teleparallel gravity, we
refer the reader to the work of Pereira and others
\cite{AldrovandiPereira, AGP, AP, Itin}.

\subsection{General relativity in teleparallel language}
\label{tel-grav}

The `coframe field' or `vielbein' is important in many approaches to
gravity, and especially in teleparallel gravity.  Locally, a coframe field 
is an $\R^n$-valued 1-form $e \maps TM \to
\R^n$, where $M$ is a $(p+q)$-dimensional manifold.  When $e$ is
nondegenerate---that is, an isomorphism when
restricted to each tangent space---it gives a metric on $M$ via pullback. 

To define coframe fields globally, even when $M$ is not parallelizable, 
we start by fixing a vector bundle $\fake \to M$
isomorphic to the tangent bundle, and equipped with a metric $\eta$ of
signature $(p,q)$.  We call $\fake$ a \define{fake tangent bundle} for
$M$.  We then define a \define{coframe field} to be a 
vector bundle isomorphism
\[ 
\xymatrix @!0
{ TM
 \ar [ddr]
  \ar[rr]^{\displaystyle e} 
  & &  
  \fake
  \ar[ddl]
  \\  \\ &
   M 
 }
\]
The coframe field lets us pull back the inner product on $\fake$ and
get an inner product on the tangent bundle, making $M$ into a
semi-Riemannian manifold with a metric $g$ of signature $(p,q)$.  

As mentioned, these ideas are useful in a number of approaches to
gravity, most notably the Palatini approach, where the fiber $\fake_x$ 
is sometimes called the `internal space'.  But in teleparallel gravity we
must go further and assume that $\fake$ is also equipped with a flat
metric-preserving connection, say $D$.  We can then use the coframe
field to pull this back to a connection on $TM$, the
\define{Weitzenb\"ock connection}.  We write this as $\D$.  By
construction, the Weitzenb\"ock connection is flat and
metric-compatible.  However, its torsion:
\[
     T(v,w) = \D_v w -\D_w v -[v,w]
\]
is typically nonzero.

The idea behind teleparallel gravity is to take familiar notions from 
general relativity and write them in terms of the coframe field $e$ and 
its Weitzenb\"ock connection $\D$ instead of the metric $g$ and its 
Levi-Civita connection $\LC$.  To do this it is useful to consider
the \define{contorsion}, which is the difference of the Weitzenb\"ock
and Levi-Civita connections:
\[     K =  \D - \LC \]
Since the Levi-Civita connection can be computed from the metric, which 
in turn can be computed from the coframe field, we can write
$K$ explicitly in terms of the coframe field $e$ and its Weitzenb\"ock
connection $\D$.  We do not need the explicit formula here: we merely
want to note how the contorsion exhibits the physical meaning of
teleparallel gravity.

For example, consider the motion of a particle in free fall.   In general 
relativity, the particle's worldline $\gamma(s)$ obeys the geodesic equation 
for the Levi-Civita connection, stating that its covariant acceleration 
vanishes:
\[
    \LC_{v(s)} v(s) = 0
\]
where $v(s) = \gamma'(s)$ is the particle's velocity.  In teleparallel
gravity, on the other hand, the equation governing a particle's motion
becomes
\[
    \D_{v(s)} v(s) =  K_{v(s)} v(s) 
\]
From this perspective, the particle accelerates away from the geodesic 
determined by the connection $\D$, and the contorsion is interpreted 
as the \textit{gravitational force} responsible for the acceleration. 

This may seem like a shell game designed to hide the `true meaning' of
general relativity.  On the other hand, it is interesting that gravity
admits such a qualitatively different alternative interpretation.
Moreover, while we are only considering the teleparallel equivalent of
general relativity here, more general versions of teleparallel gravity
make physical predictions different from those of general relativity
\cite{AP,Itin}.

In any case, using the same strategy to rewrite the Einstein--Hilbert 
action of general relativity, one obtains, up to a boundary term, 
the teleparallel gravity action \cite{AP,Maluf}, which is proportional to:
\begin{equation}
\label{tele-action}
 S[e] = \int d^nx\; \det(e) \left(
     \frac14 T^\rho{}_{\mu\nu} T_\rho{}^{\mu\nu} 
     + \frac12 T^\rho{}_{\mu\nu} T^{\nu\mu}{}_\rho 
     - T_{\rho\mu}{}^\rho T^{\nu\mu}{}_\nu
   \right)
\end{equation}
Here the components of the torsion in a coordinate frame $x^\mu$ 
are given by
\[
     T(\partial_\mu,\partial_\nu) = T^\rho{}_{\mu\nu} \partial_\rho ,
\]
and indices are moved using the metric pulled back along $e$.
We have written $d^n x \det(e)$ for the volume form, where
$n=p+q$ is the dimension of spacetime.  Of course, one usually
considers the case of 3+1 dimensions, but the action can be written in
arbitrary dimension and signature.  Note also that while the action
involves the torsion $T$, this is a function of the Weitzenb\"ock
connection, which is a function of the coframe field, so the action is
a function only of $e$.

To obtain the field equations, one can vary the action $S[e]$ with
respect to the coframe field $e$.  Of course, we can also `cheat',
using the known answer from general relativity.  We simply convert
$S[e]$ back into the usual Einstein--Hilbert action, and then recast
the resulting Einstein equations in terms of the coframe and
Weitzenb\"ock connection.  We shall not do this here, since we make no
use of these equations; we refer the reader to references already
cited in this section for more details.

\subsection{Torsion and the coframe field}
\label{torsion}

Since torsion is less widely studied than curvature---many geometry
and physics books set torsion to zero from the outset---we recall some
ideas about it here.  We take an approach suited to teleparallel
gravity, but also to our interpretation of it in terms of the 
Poincar\'e 2-group.  

We start with a manifold $M$ equipped with a fake tangent bundle
$\fake \to M$.  From this we can build a principal $\O(p,q)$
bundle $\ff \to M$, called the \define{fake frame bundle}.  The idea is
to mimic the usual construction of the frame bundle of a
semi-Riemannian manifold.  Thus, we let the fiber $\ff_x$ at a point
$x \in M$ be the space of linear isometries $\R^{p,q} \to \fake_x$.  The
group $\O(p,q)$ acts on each fiber in an obvious way, and we can check
that these fibers fit together to form the total space $\ff$ of a
principal $\O(p,q)$ bundle over $M$.

Next, suppose we have a coframe field
\[               e\maps TM \stackrel{\sim}{\longrightarrow} \fake  .\]
This allows us to pull back the inner product on $\fake$ and get 
an inner product on the tangent bundle, making $M$ into a semi-Riemannian 
manifold with a metric $g$ given by
\[            g(v,w) = \eta(e(v), e(w)) .\]

Next, suppose we have a flat connection $\om$ on $\ff$.  This determines
a flat metric-preserving connection on $\fake$, say $D$.  We can then
pull this back using the coframe field to obtain the Weitzenb\"ock
connection $\D$ on $TM$.  Explicitly, the \define{Weitzenb\"ock connection} is 
determined by requiring that
\[             e( \D_v w) = D_v (e(w)) \]
for all vector fields $v,w$ on $M$.  We note that this condition is invariant
under gauge transformations of $\ff$. For, such a transformation gives a 
section $g$ of ${\rm End}(\fake)$, acting on the coframe by 
$e \mapsto ge$ and commuting with the connection $D$, so that:
\[
     ge( \D_v w) = g\,D_v (e(w)) =  D_v (ge(w)).
\]
Hence the Weitzenb\"ock connection itself is gauge invariant.

The torsion of the Weitzenb\"ock connection is the $TM$-valued 2-form
given by
\[
     T(v,w) = \D_v w -\D_w v -[v,w]
\]
On the other hand, the coframe field $e$ can be seen as a $\fake$-valued
1-form, so its covariant exterior derivative is a $\fake$-valued 2-form 
$d_\om e$.  Crucially, this is just the torsion in disguise.  More precisely,
we can translate between $d_\om e$ and $T$ using the coframe field:
\[        e(T(v,w)) = (d_\om e)(v,w)   \]
for all vector fields $v,w$ on $M$.  To see this recall that the covariant 
exterior derivative $d_\om$ is defined just like the ordinary differential 
$d$, but where all directional derivatives are replaced 
by covariant derivatives.  Hence:
\begin{align}
\nonumber
(d_\om e)(v,w) :&=D_v(e(w)) - D_w(e(v)) - e([v,w]) 
\\
\label{internal-torsion}
&=e\left(\D_v w - \D_w v - [v,w]\right)
\\
\nonumber
&= e(T(v,w)). 
\end{align}

There is also a simple relationship between the notions of torsion and
contorsion.  If $\tilde \om$ is the torsion-free connection for the
coframe field $e$, then we find
\begin{align*}
d_\om e &= d_\om e - d_{\tilde \om}e 
\\
&=(de + \om\we e) - (de + \tilde \om \we e)
\\
&= K \we e
\end{align*}
where $K = \om - \tilde \om$ is the contorsion.  

\subsection{Poincar\'e 2-connections}
\label{Poincare 2-connections}

Now we reach our first main result, a relationship between
flat connections with torsion and Poincar\'e 2-connections.  As we saw
in Proposition \ref{2-frame}, any semi-Riemannian manifold comes
equipped with a principal $\Po(p,q)$ 2-bundle, its `2-frame 2-bundle'.
There is also a `fake' version of this construction when our manifold
$M$ is equipped with a fake tangent bundle $\fake \to M$:

\begin{defn} 
If $\fake \to M$ is a fake tangent bundle with a metric of signature
$(p,q)$, and $\ff \to M$ is its fake frame bundle, then we define the
\define{fake 2-frame 2-bundle} to be the principal $\Po(p,q)$ 
2-bundle $\tff \to M$ where
\[   \tff = \ff \times_{\O(p,q)} \Po(p,q)  . \]
\end{defn}  

As we have seen, the raw ingredients of teleparallel gravity are a
flat connection $\om$ on the fake frame bundle, together with a
$\fake$-valued 1-form $e$.  The torsion, which plays a crucial role in
teleparallel gravity, can then be reinterpreted as the $\fake$-valued
2-form $d_\om e$.  Our result is that the pair $(\om, d_\om e)$ is then
a flat 2-connection on the fake 2-frame 2-bundle.  Conversely, if some
topological conditions hold, every such flat 2-connection arises this 
way:

\begin{thm}
\label{thm:main}
Suppose $\fake \to M$ is a fake tangent bundle with a metric of signature
$(p,q)$.   If $\om$ is a flat connection on the fake frame bundle of $M$
and $e$ is a $\fake$-valued 1-form, then $(\om, d_\om e)$ is a flat 
2-connection on the fake 2-frame 2-bundle $\tff$.  Conversely, if $M$ is simply 
connected and has vanishing 2nd de Rham cohomology, every flat 
2-connection on $\tff$ arises in this way.
\end{thm}

\begin{proof}
By construction, $\tff$ has $\ff$ as its manifold of objects.  Thus,
by Proposition \ref{global-2conn}, a 2-connection on $\tff$ amounts to
a connection $\om$ on $\ff$ together with a 2-form on $M$ with values in
$\fake = (\ff \times_{\O(p,q)} \R^{p,q})$, satisfying fake flatness.
For the Poincar\'e 2-group, $\dt = 0$, so fake flatness simply means
the connection $\om$ is flat.  We thus get a 2-connection $(\om,d_\om e)$ on
$\tff$ from any flat connection $\om$ on $\ff$ and $\fake$-valued 1-form
$e$.

To see that the 2-connection $(\om, d_\om e)$ is flat, we must also check
that its 2-curvature vanishes.  This 2-curvature is just the exterior 
covariant derivative $d_\om (d_\om e)$.  However, by the Bianchi identity 
\[
     d_\om (d_\om e) = R \we e 
\]
where $R$ is the curvature of $\om$, and $R=0$ because $\om$ is flat.
It follows that the 2-connection $(\om, d_\om e)$ is flat.

For the converse, suppose we have any flat 2-connection $(A,B)$ on 
$\tff$; we will bring in the additional topological assumptions 
as we need them. As just discussed, the fake flatness condition 
implies that $A=\om$ is a {\em flat} connection on $\ff$.  
By Proposition~\ref{global-2conn}, $B$ is a $\fake$-valued 2-form, and 
since the 2-curvature vanishes, $B$ is covariantly closed:
\[
   d_\om B = 0.
\]
Now for any $\fake$-valued form $X$ we have $(d_\om)^2 X = R\we X = 0$,
since $\om$ is flat.  We thus get a cochain complex of $\fake$-valued forms:
\[
\xymatrix{
     \Om^0(M,\fake) \ar[r]^{d_\om} 
     & \Om^1(M,\fake) \ar[r]^{d_\om} 
     & \Om^2(M,\fake) \ar[r]^{\;\;\quad d_\om} 
     & \cdots
}
\]
Let us denote the cohomology of this complex by $H_\om^\bullet
(M,\fake)$.  If we had $H_\om^2 (M,\fake) = 0$, then $d_\om B= 0$ would
imply $B = d_\om e$ for some $e \in \om^1(M,\fake)$.  This would
complete the proof.

We now begin imposing topological conditions to guarantee that $H_\om^2
(M,\fake) = 0$.  First, suppose that (each component of) $M$ is simply
connected.  In this case, note that $FM$ has a flat connection if and
only if $M$ is parallelizable, meaning that the frame bundle
$FM$ is trivializable.  For, any trivialization of $FM$ determines a
unique connection such that parallel transport preserves this
trivialization.  And conversely, given a flat connection, $\pi_1(M)=
1$ implies parallel transport is completely path independent, so we
can trivialize $FM$ by parallel translation, starting from a frame at
one point in each component of $M$.

Since $\ff$ is isomorphic to $FM$, the same is true for it: if
$M$ is simply connected, $\ff$ admits a flat connection if and only
if $M$ is parallelizable.

Thus we may assume $M$ is parallelizable.  By trivializing the fake frame
bundle of $M$, we may think of the 2-form $B$ as taking values not in
$\fake$, but simply in $\R^{p,q}$.  We are thus reduced to a cochain
complex of the form
\[
\xymatrix{
     \Om^0(M,\R^{p,q}) \ar[r]^{d_\om} 
     & \Om^1(M,\R^{p,q}) \ar[r]^{d_\om} 
     & \Om^2(M,\R^{p,q}) \ar[r]^{\;\;\quad d_\om} 
     & \cdots
}
\]
and $B$ determines a class in $H^2_\om(M,\R^{p,q})$.

Gauge transformations act on the cohomology $H_\om^\bullet$ in a natural
way.  Having trivialized $\ff$, a gauge transformation may be written $\om'
= g\om g^{-1} +gdg^{-1}$.  Using this formula, it is easy to show that
\[
   d_{\om'} (gX) %= d(gX) + g\om g^{-1}\we gX - (dg) g^{-1} \we gX 
   = g(d_\om X).
\]
Thus, if $Y$ and $Y'$ are cohomologous for $\om$, then $Y- Y' = d_\om X$,
hence $gY-gY' = d_{\om'}(gX)$, so that $gY$ and $gY'$ are cohomologous
for ${\om'}$.  Since $M$ is simply connected, all flat connections on
$\ff$ are gauge equivalent.  In particular, they are all gauge
equivalent to the `zero connection' in our chosen trivialization of
$\ff$.

In this gauge, the covariant differential $d_\om$ becomes the ordinary
differential $d$, and the question is simply whether $dB = 0$ implies
$B = de$ for some $e$.  This is essentially a question of de Rham
cohomology: thinking of $B\in \Om^2(M,\R^{p,q})$ as a $(p+q)$-tuple
of 2-forms $B_1,\ldots, B_{p+q}$, it is clear that $B$ is exact
precisely when each $B_{i}$ is.  So, demanding now that the second de
Rham cohomology of $M$ vanishes, we get $B=de$ for some coframe field
$e$, in the chosen gauge.  Switching back to an arbitrary gauge, this
shows that our 2-connection is really of the form $(\om,d_\om e)$.
\end{proof}

In particular, we can take the fake tangent bundle of $M$ to be the
actual tangent bundle:

\begin{cor} 
\label{corollary}
Suppose $M$ is a manifold equipped with a metric $g$ of signature
$(p,q)$.  If $\omega$ is a flat metric-compatible connection on the
frame bundle of $M$, then the pair $(\omega,T)$, where $T$ is the
torsion of $\omega$, is a flat strict 2-connection on the 2-frame
2-bundle of $M$.
\end{cor}
 
\begin{proof}  
The tangent bundle $TM$, equipped with the metric $g$, is a 
particular case of a fake tangent bundle on $M$, and the corresponding 
fake frame bundle is just the usual (orthonormal) frame bundle.  Taking 
the coframe field $e\maps TM\to TM$ to be the identity map, we 
have $\D = D$, so the calculation (\ref{internal-torsion}) simplifies
to
\[
   (d_\omega e)(v,w) = \D_v w - \D_w v - [v,w] = T(v,w).
\]
Thus $d_\omega e = T$, and by the previous theorem, $(\omega,T)$ is a
flat 2-connection on the fake 2-frame 2-bundle, which here is just
the original 2-frame 2-bundle of Def.~\ref{2FM}.
\end{proof}

While the converse in Theorem \ref{thm:main} assumes the manifold
is simply connected, this condition is not really necessary.   We only
introduced it to make it easy to check that the cohomology group
$H_\om^2(M,\fake)$ vanishes.  Our proof gives more:

\begin{cor}
\label{corollary2}
Suppose $\fake \to M$ is a fake tangent bundle with a metric of signature
$(p,q)$.   Suppose $(A,B)$ is a flat 2-connection on $\tfm$.  Then $A=\om$
is a flat connection on $\fake$, and if $H_\om^2(M,\fake) = 0$, then 
$B = d_\om e$ for some $\fake$-valued 1-form $e$ on $M$.
\end{cor}

\begin{proof} 
This was established in the proof of Theorem \ref{thm:main}. 
\end{proof}

\section{Teleparallel gravity and the teleparallel 2-group}
\label{teleparallel-2-group}

In the previous section, we showed that the Poincar\'e 2-group is
related to `teleparallel geometry'---the geometry of flat connections
with torsion.  However, the relevance to teleparallel {\em gravity} is
not yet clear.  The main field in the action for teleparallel gravity
is the {\em coframe field}, while the Poincar\'e 2-connection 
$(\om,d_\om e)$ seems to play more of a supporting role.    It
would be much more satisfying if we could view the coframe field in
terms of a 2-connection as well.  In fact we can do just this if we 
extend the Poincar\'e 2-group to a larger 2-group: the `teleparallel 
2-group'.  

\subsection{The teleparallel 2-group}

Let us start by defining this 2-group directly:
\begin{defn}
\label{defn:Tele}
The \define{teleparallel 2-group} is the Lie 2-group \define{$\Tel(p,q)$} 
for which:
\begin{itemize}
\item $\IO(p,q)$ is the Lie group of objects,
\item $\IO(p,q) \ltimes \R^{p,q}$ is the Lie group of morphisms, where
the semidirect product is defined using the following action:
\[      \begin{array}{ccl}
\alpha \maps \IO(p,q) \times \R^{p,q} &\to& \R^{p,q}   \\
              ((g,v), w) &\mapsto& g w  ,
\end{array}
\]
where $(g,v) \in \IO(p,q) = \O(p,q) \ltimes \R^{p,q}$ and 
$w \in \R^{p,q}$.
\item the source of the morphism $((g,v),w)$ is $(g,v)$,
\item the target of the morphism $((g,v),w)$ is $(g,v+w)$
\item the composite $((g,v),w) \circ ((g',v'),w')$, when defined, 
is $((g',v'),w+w')$
\end{itemize}
\end{defn}
\noindent Explicitly, the product in the group of morphisms is given by
\begin{align*}
  ((g,v), w)  ((g',v'), w') %&= ((g,v)(g',v'),w + \alpha(g,v)w') \\
    &= ((gg',v+gv'),w+gw').
\end{align*}

The Lie crossed module corresponding to $\Tel(p,q)$ has:
\begin{itemize}
\item $G = \IO(p,q)$
\item $H = \R^{p,q}$
\item $t\maps \R^{p,q} \to \IO(p,q)$ the inclusion homomorphism
$v\mapsto(1,v)$
\item $\alpha(g,v)w = gw$.
\end{itemize}
\noindent The corresponding Lie 2-algebra, the \define{teleparallel
Lie 2-algebra}, has an infinitesimal crossed module with
\begin{itemize}
\item $\g = \io(p,q)$
\item $\h = \R^{p,q}$
\item $\dt\maps \R^{p,q} \to \io(p,q)$ the inclusion homomorphism
$v\mapsto(0,v)$
\item $\dalpha(\xi,v)w = \xi w$.
\end{itemize}

Another way to think about the teleparallel 2-group starts with the
Poincar\'e group.  Given any Lie group $G$ there is a Lie 2-group
$\E(G)$ with $G$ as its group of objects and a unique morphism between
any two objects.  This corresponds to the crossed module $1 \maps G
\to G$, where $G$ acts on itself by conjugation.  This Lie 2-group
plays an important role in topology, as first noted by Segal
\cite{BCSS,Segal}: it is closely related to the contractible $G$ space
$EG$ that is the total space of the universal principal $G$-bundle $EG
\to BG$.  However, what matters to us here is that teleparallel 2-group
is a sub-2-group of $\E(\IO(p,q))$.

\begin{defn} 
Given a Lie group $G$, define \define{$\E(G)$} to be the Lie 2-group 
for which:
\begin{itemize}
\item $G$ is the Lie group of objects,
\item $G \ltimes G$ is the Lie group of morphisms, where $G$ acts on
itself by conjugation,
\item the source of the morphism $(g,h) \in G \ltimes G$ is $g$,
\item the target of the morphism $(g,h) \in G \ltimes G$ is $h g$
\item the composite $(g,h) \circ (g',h')$, when defined, is $(g', h h')$ 
\end{itemize}
\end{defn}

It is worth noting that $G \ltimes G$ is isomorphic to $G \times G$.
However, the description of $\E(G)$ using a semidirect project makes
it easier to understand the inclusion
\[      i \maps \Tel(p,q) \to \E(\IO(p,q))  .\]
Namely, the inclusions at the object and morphism levels, which we
denote as $i_0$ and $i_1$, are the obvious ones:
\[    \begin{array}{cccl}
i_0 \maps & \IO(p,q) &\to& \IO(p,q)  \\
i_1 \maps & \IO(p,q) \ltimes \mathbb{R}^{p,q} &\to& \IO(p,q) \ltimes \IO(p,q) .
\end{array} \]

The 2-group $\E(G)$ can also be obtained from the crossed module
$(G,H,t,\alpha)$ for which $G$ and $H$ are the same group $G$, $t\maps
G\to G$ is the identity homomorphism, and $\alpha$ is the action
of $G$ on itself by conjugation.

There is also an inclusion of Poincar\'e 2-group in the teleparallel
2-group:
\[    j \maps \Po(p,q) \to \Tel(p,q)  \]
where the inclusions at the object and morphism levels: 
\[    \begin{array}{cccl}
j_0 \maps & \O(p,q) &\to& \IO(p,q)  \\
j_1 \maps & \O(p,q)\ltimes \R^{p,q} &\to& \IO(p,q) \ltimes \R^{p,q} 
\end{array} \]
are the obvious ones.  Here we have written the group of morphisms of
$\Po(p,q)$ as $ \O(p,q)\ltimes \R^{p,q}$ rather than $\IO(p,q)$ to
emphasize that $j_1$ maps the $\mathbb{R}^{p,q}$ into the second
factor of $\IO(p,q) \ltimes \R^{p,q}$, not
into the copy of $\mathbb{R}^{p,q}$ hiding in the first factor.

Remember that we can regard the group $\O(p,q)$ as a Lie 2-group whose
morphisms are all identities.  As pointed out by Urs Schreiber, this
2-group is equivalent in the 2-category of Lie 2-groups, though not
isomorphic, to $\Tel(p,q)$.  The obvious inclusion $\O(p,q) \to
\Tel(p,q)$ is an equivalence.  In `weak' higher gauge theory,
equivalent 2-groups are often interchangeable.  Nonetheless, $\O(p,q)$
and $\Tel(p,q)$ play different roles in our work.  This remains
somewhat mysterious and deserves further exploration.

The 2-groups that play a role in this paper can be summarized as follows.
Each is included in the next:
\[
\begin{array}{l|ccccc}
\G &  &\O(p,q) & \Po(p,q) & \Tel(p,q) & \E(\IO(p,q))  \\[.2em] \hline \\[-.8em]
\G_0  && \O(p,q) &\O(p,q) & \IO(p,q) & \IO(p,q) \\
\G_1  && \O(p,q) & \O(p,q)\ltimes \R^{p,q} & \IO(p,q)\ltimes \R^{p,q} & \IO(p,q)\ltimes \IO(p,q)
\end{array}
\]

\subsection{Poincar\'e connections}

Our next goal is to describe 2-connections on certain 2-bundles with
the teleparallel 2-group as gauge 2-group.  Such a 2-bundle has a
bundle of objects that is a principal Poincar\'e group bundle.
Thus, a 2-connection on a principal $\Tel(p,q)$ 2-bundle consists partly 
of a Poincar\'e connection.  It will therefore be helpful to have a 
few facts about Poincar\'e connections at hand.

First, if we pull the adjoint representation of the Poincar\'e 
group $\IO(p,q)$ back to the Lorentz group $\O(p,q)$, it splits 
into a direct sum 
\[
    \io(p,q) = \o(p,q) \oplus \R^{p,q}
\]
of irreducible $\O(p,q)$ representations.  Since a connection $A$ on a
principal $\IO(p,q)$ bundle can locally be described by an $\io(p,q)$-valued
1-form, we can split it as:
\[
A = \om + e,
\]
where $\om$ is $\o(p,q)$-valued and $e$ is $\R^{p,q}$-valued.  Note
that $\om$ can be seen as a Lorentz connection, while $e$ can be seen
as a coframe field, at least when it restricts to an isomorphism on
each tangent space.  It is also easy to check that the same direct sum
of representations splits the curvature of $A$, $F = dA + A \we A$,
into two parts as follows:
\[
F = R + d_\om e.
\]
Here the $\o(p,q)$-valued part 
\[    R = d\om + \om \wedge \om  \]
is the curvature of $\om$, while the $\R^{p,q}$-valued part 
is the torsion
$   d_\om e $.

To see how this all works globally, it helps to note that the coframe
field, defined in Section \ref{tel-grav} as an isomorphism $TM \to
\fake$, can equivalently be viewed as a certain kind of 1-form on
$\ff$:
\begin{lemma}
\label{coframe-lemma}
Let $\fake$ be a fake tangent bundle on $M$, and $p\maps \ff\to M$ the
corresponding fake frame bundle.  Then there is a canonical one-to-one
correspondence between:
\begin{itemize}
\item vector bundle morphisms $e\maps TM \to \fake$, and
\item $\R^{p,q}$-valued 1-forms $\varepsilon$ on $\ff$ that are:
 \begin{itemize}
 \item horizontal: $\varepsilon$ vanishes on $\ker(dp)$
 \item $\O(p,q)$-equivariant: $R^*_h \varepsilon = h^{-1} \circ
 \varepsilon$ for all $h\in \O(p,q)$.
 \end{itemize}
\end{itemize}
Moreover, the first of these is an isomorphism precisely when the
second is nondegenerate, meaning that each restriction $\varepsilon
\maps T_f\ff \to \R^{p,q}$ has rank $p+q$.

\end{lemma}

\begin{proof}
Suppose $e\maps TM \to \fake$ is a vector bundle morphism.  Let $\pi$
be the projection map for the tangent bundle of $\ff$:
\[
    \pi\maps  T\ff \to \ff.
\]
Then we can define $\tilde e \maps T\ff \to \R^{p,q}$, an
$\R^{p,q}$-valued 1-form on $\ff$, as:
\[
   \tilde e(v) = \pi(v)^{-1}\big(e\circ dp(v)\big)\qquad \forall v\in T\ff.
\]
where $dp\maps T\ff \to TM$ is the differential of the fake frame
bundle $p\maps \ff\to M$.  In this formula we are using the fact that
$\pi(v)\in\ff$ is itself an isomorphism $\R^{p,q} \to
\fake_{p(\pi(v))}$, by the definition of the fake frame bundle in
Section \ref{torsion}, and taking the inverse of this.  This $\tilde
e$ is horizontal, since it obviously vanishes on the kernel of $dp$.
All that remains to check is equivariance.  This follows simply from
equivariance of $\pi \maps T\ff \to \ff$ and $dp\maps T\ff \to TM$,
where the action $\O(p,q)$ on $TM$ is trivial, induced from the
trivial action on $M$:
\begin{align*}
R^*_h\tilde e(v) = 
\tilde e(R_{h*} v) &= \pi(R_{h*} v)^{-1}\big(e\circ dp(R_{h*} v)\big) \\ 
&= h^{-1} \circ \pi(v)^{-1} (e\circ dp(v)) \\ 
&= h^{-1} \circ \tilde e(v) .
\end{align*}

Conversely, suppose we are given a horizontal and
equivariant 1-form $\varepsilon\maps T\ff \to \R^{p,q}$, and let us
construct a vector bundle morphism $\bar \varepsilon\maps TM \to
\fake$.  Given $w \in T_xM$, pick any $f\in \ff_x$ and $v\in T_f\ff$.
Since $\varepsilon$ is horizontal, $\varepsilon(v)$ does not depend on
which $v\in T_f\ff$ we pick; it does depend on the point $f$ in the
fiber, but this dependence is $\O(p,q)$-equivariant.  Using this, it
is easy to check that demanding
\[
     \bar \varepsilon(w) = \pi(v)\big(\varepsilon(v)\big) \quad
     \forall w\in TM, v\in T\ff \text{ with } dp(v)=w
\]
uniquely determines a vector bundle morphism 
$\bar \varepsilon \maps TM \to \fake$.  

It is also easy to see that these processes are inverse, namely that
$\bar {\tilde e} = e$ and $\tilde {\bar \varepsilon} = \varepsilon$.
Finally, it is easy to check that the vector bundle morphism $e\maps
TM \to \fake$ is an isomorphism precisely when the corresponding
horizontal equivariant 1-form $\varepsilon \maps T\ff \to \R^{p,q}$ is
nondegenerate.
\end{proof}

Every manifold equipped with a fake tangent bundle $\fake \to M$ has a
principal Poincar\'e group bundle over it.  To build this, we start
with the fake frame bundle $\ff \to M$ as described in Section
\ref{torsion}.  This is a principal Lorentz group bundle.  We then
extend this to a principal Poincar\'e group bundle over $M$, the
\define{extended fake frame bundle $I\ff$}:  
\[
     I\ff := \ff \times_{\O(p,q)} \IO(p,q).    
\]
A connection on the extended fake frame bundle is an $\io(p,q)$-valued 
1-form $A$ on the total space $I\ff$ satisfying the usual equations.
But such connections have a more intuitive description:

\begin{prop}
\label{iopq-conn}
There is a canonical $\O(p,q)$-equivariant correspondence between the
following kinds of data:
\begin{itemize}
\item a connection $A$ on the extended fake frame bundle $I\ff$, with
nondegenerate $\R^{p,q}$ part;
\item a connection $\om$ on the fake frame bundle $\ff$ together with
a coframe field $e$.
\end{itemize}
\end{prop}
\begin{proof}
A connection $A$ on $I\ff\to M$ is an $\io(p,q)$-valued 1-form on the
total space satisfying the usual equivariance properties under
$\IO(p,q)$.  This connection pulls back to a 1-form on $\ff$, and the
direct sum of $\O(p,q)$ representations $\io(p,q) = \o(p,q) \oplus
\R^{p,q}$ splits this into 1-forms $\om$ and $e$ with values in
$\o(p,q)$ and $\R^{p,q}$.  One can check that $\om$ is a connection on
$\ff$ and that when $e$ nondegenerate, it is equivalent to a coframe
field, via Lemma~\ref{coframe-lemma}.

Conversely, a connection $\om$ on $\ff$ and a coframe field, viewed as
an equivariant 1-form $e:T\ff \to \R^{p,q}$, assemble into an
$\io(p,q)$-valued 1-form on $\ff$, thanks to our direct sum of
$\O(p,q)$ representations.  This 1-form has a unique equivariant
extension to the associated bundle $I\ff=\ff
\times_{\O(p,q)}\IO(p,q)$, and this is a connection $A$.  The
$\R^{p,q}$ part of $A$ is nondegenerate because $e$ is.

This correspondence is clearly equivariant under gauge transformations
of the principal $\O(p,q)$ bundle $\ff$.
\end{proof}

It follows that the local description of the curvature of an
$\IO(p,q)$ connection, given at the beginning of this section, holds
globally as well:
 
\begin{prop}
\label{iopq-curv}
Let $A$ be a connection on the extended fake frame bundle $\I\ff$
described by a connection $\om$ on $\ff$ and a coframe field $e$.
Then the curvature $F$ of $A$ consists of the curvature $R = d\om
+ \om \we \om$ of $\om$ and the torsion $d_\om e$.
\end{prop}
 
\subsection{Teleparallel 2-connections} 

We now have all the ingredients we need to build $\Tel(p,q)$
2-connections.  The first step is to note that any manifold $M$
equipped with a fake tangent bundle has a principal $\Tel(p,q)$
2-bundle over it.  We can build this by forming the fake frame bundle
and then extending its gauge 2-group from the Lorentz group to the
teleparallel 2-group:

\begin{defn} 
If $\fake \to M$ is a fake tangent bundle with a metric of signature
$(p,q)$, and $\ff \to M$ is its fake frame bundle, then we define the
\define{teleparallel 2-bundle} to be the principal $\Tel(p,q)$ 
2-bundle $\Tel(\ff) \to M$ where
\[   \Tel(\ff) =  \ff \times_{\O(p,q)} \Tel(p,q).\]
\end{defn}  

\noindent Alternatively, we can start with the fake 2-frame 2-bundle 
and extend its gauge 2-group from the Poincar\'e 2-group to
the teleparallel 2-group.  This gives a 2-bundle
\[  
\tff \times_{\Po(p,q)} \Tel(p,q)
\] 
which is canonically isomorphic to $\Tel(\ff)$.  So, we are free to
also think of this as the teleparallel 2-bundle.

To describe 2-connections on the teleparallel 2-bundle, it is perhaps
easiest to go one step further and extend the gauge 2-group
all the way to $\E(\IO(p,q))$.  More generally, one can describe $\E(G)$
2-connections for an arbitrary Lie group $G$:
\begin{prop}
\label{e(g)-conn}
Let $P$ be a principal $G$-bundle, $\E(P)=P \times_G \E(G)$ the
associated $\E(G)$ 2-bundle.  For each connection $A$ on $P$ there is
a unique strict 2-connection on $\E(P)$ whose 1-form part is $A$; the
2-form part is the curvature $F= dA + A\we A$.
\end{prop}
\begin{proof}
The bundle of objects of $\E(P)$ is $P$, so the 1-form part of a
2-connection $(A,B)$ is a connection on $P$.  The fake flatness
condition in this case simply says $B=F$, since `$\dt$' in the
infinitesimal crossed module is the identity map.
\end{proof}

As a corollary, an $\E(\IO(p,q))$ 2-connection consists of four
parts: connection, coframe, curvature, and torsion.  However, the first 
two parts determine the others:
\begin{prop}
Let $\ff$ be a fake frame bundle.  A 2-connection on
$\ff\times_{\O(p,q)} \E(\IO(p,q))$ is specified uniquely by a
connection $\om$ on $\ff$ and a $\fake$-valued 1-form $e$.  The 2-form
part of this 2-connection consists of the curvature $R= d\om +
\om\we\om$ together with the torsion $d_\om e$.
\end{prop}
\begin{proof}
$\ff\times_{\O(p,q)} \E(\IO(p,q))$ is canonically isomorphic to
$\E(\I\ff)$, where $\I\ff$ is the extended fake frame bundle.  In light
of Propositions~\ref{iopq-conn} and \ref{iopq-curv}, the result then
follows immediately from Proposition~\ref{e(g)-conn}.
\end{proof}

Now consider $\Tel(p,q)\subseteq \E(\IO(p,q))$.  A $\Tel(p,q)$
2-connection can be thought of as a special case of an $\E(\IO(p,q))$
2-connection.  We can still interpret it as an $\O(p,q)$ connection
$\om$ and coframe $e$, together with the torsion $d_\om e$.  However,
fake flatness now implies that $R = d\om + \om\we\om = 0$, so $\om$
must be a {\em flat} connection.  Summarizing:

\begin{thm}
\label{telpq-conn}
Let $M$ be a manifold equipped with a fake tangent bundle. 
A 2-connection on the $\Tel(p,q)$ 2-bundle
\[
    \ff \times_{\O(p,q)} \Tel(p,q)
\]   
consists of:   
\begin{itemize}
\item a {\em flat} connection $\om$ on $\ff$
\item a $\fake$-valued 1-form $e$, and 
\item the $\fake$-valued 2-form $d_\om e$.
\end{itemize}
\end{thm}
\noindent In the case where $e$ is an isomorphism, it can be viewed as
a coframe field, and we have just the fields needed for describing
teleparallel geometry.

From this perspective, $\Tel(p,q)$ results from $\E(\IO(p,q))$ by
truncating the part where the curvature lives, leaving only torsion at
the morphism level.  Crucially, this truncation does not simply
`forget' the curvature part: the fake flatness condition forces the
curvature to vanish.

The reader may well wonder if we could do an analogous truncation that
would allow us to use 2-connections to describe geometries that are
not flat but rather torsion-free, removing the $\O(p,q)$ part of the
group of morphisms in the 2-group, rather than the $\R^{p,q}$ part.
In the language of crossed modules, this would mean taking the group
$H$ to be $\O(p,q)$.  But unlike $\R^{p,q}$, $\O(p,q)$ is not a normal
subgroup of $\IO(p,q)$, so the action of $\IO(p,q)$ does not restrict
to an action on this subgroup.  Thus, this strategy fails to give a
2-group.

Something interesting happens when we calculate the curvature of a
$\Tel(p,q)$ 2-con\-nec\-tion:

\begin{prop}
\label{2flat}
Every strict 2-connection on a strict $\Tel(p,q)$ 2-bundle has
vanishing 2-curvature.
\end{prop}
\begin{proof} 
Let $(A,B) = ((\om,e),d_\om e)$ be a 2-connection on a trivial principal
$\Tel(p,q)$ 2-bundle.  A direct calculation of the 2-curvature 3-form
$G$ gives
\begin{align*}
  G &= d_A (d_\om e) \\ 
    %&= d(d_\om e) + A \we d_\om e \\
    %&= d(d_\om e) + \om \we d_\om e + e\we d_\om e \\
    %&= d_\om^2 e  + e\we d_\om e \\
    &= R\we d_\om e + e \we d_\om e.
\end{align*}
The first term here is zero, since $R=0$.  The second term is also
zero: the wedge product really involves the bracket of Lie algebra
valued forms, where $e$ and $d_\om e$ both live in the abelian
subalgebra $\R^{p,q}\subseteq \io(p,q)$.  The condition $G=0$ is a
local condition, so this proves any 2-connection has vanishing
2-curvature, on nontrivial 2-bundles as well.
\end{proof}

This phenomenon can also be seen by noting that the holonomy of a
teleparallel 2-con\-nec\-tion does not change as we deform the
surface.  Suppose we have a pair of paths $\gamma_1, \gamma_2$,
bounding a surface $\Sigma$:
\tikzstyle arrowstyle=[scale=1.3]
\tikzstyle middlearrow=[postaction={decorate,decoration={markings,
    mark=at position .52 with {\arrow[arrowstyle]{stealth'}}}}]
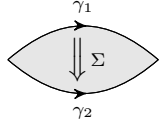
\[
\begin{tikzpicture}
\filldraw [fill=gray!20] (0,0) .. controls (.7,.6) and (1.3,.6) .. (2,0)  .. controls (1.3,-.6) and (.7,-.6) .. (0,0);
\draw [middlearrow] (0,0) .. controls (.7,.6) and (1.3,.6) .. (2,0);
\draw [middlearrow] (0,0) .. controls (.7,-.6) and (1.3,-.6) .. (2,0);
\node [scale=1.1,rotate=270] at (.9,0) {$\implies$};
\node [vsmall] at (1.2,0) {$\Sigma$};
\node [vsmall] at (1,.7) {$\gamma_1$};
\node [vsmall] at (1,-.7) {$\gamma_2$};
\end{tikzpicture}
\]
The general recipe for computing surface holonomies simplifies when we
have a $\Tel(p,q)$ 2-connection, because $\om$ is flat and the group
$\R^{p,q}$ is abelian.  Since $\om$ is flat, we can pick a local
trivialization of the fake tangent bundle for which $\om = 0$.  Using
this trivialization, the 2-form $d_\om e$ can be interpreted as an
$\R^{p,q}$-valued 2-form.  To obtain the surface holonomy, we simply
integrate this 2-form over $\Sigma$.  Stokes' theorem implies that
this surface holonomy can be rewritten as
\[
\int_\Sigma de = \int_{\gamma_2} e - \int_{\gamma_1} e.
\]  
Geometrically, $\int_{\gamma} e$ is just the `translational holonomy'
of the Poincar\'e connection along $\gamma$.  So, the surface holonomy
of a teleparallel 2-connection 
simply measures the difference between the translational holonomies
along the two bounding edges.  In particular, the surface holonomy
does not change as we apply a smooth homotopy to $\Sigma$ while
keeping its edges $\gamma_1$ and $\gamma_2$ fixed, so its 2-curvature
$G$ must vanish.

For a flat $\Po(p,q)$ 2-connection, the surface holonomy  can likewise be
seen as measuring the difference between translational holonomies
along its two bounding edges.  An immediate corollary of
Theorem~\ref{thm:main} is that, under the conditions of that theorem
(or, in any case, locally), a flat $\Po(p,q)$ 2-connection extends to
$\Tel(p,q)$ 2-connection, simply by adjoining a coframe field.  This
coframe is unique up to a covariantly closed 1-form, and adding a
covariantly closed 1-form clearly does not change the 2-holonomy.

\subsection{Teleparallel gravity as a Tel(\textit{p,q}) higher gauge theory}
\label{gravity}

We can now view teleparallel gravity as a higher gauge theory with
gauge 2-group $\Tel(p,q)$.  Let us describe how this works, beginning
with a brief review of the geometric framework we have built up.

We start with a manifold $M$ equipped with a fake tangent bundle $\fake$
and its corresponding fake frame bundle $\ff$.  From this, we build
the principal $\Tel(p,q)$ 2-bundle
\[
     \Tel(\ff) = \ff \times_{\O(p,q)} \Tel(p,q). 
\]
A 2-connection on $\Tel(\ff)$, by Theorem~\ref{telpq-conn}, is
equivalent to a flat connection $\fc$ on $\ff$ together with a
$\fake$-valued 1-form $e$.   We confine attention to $\Tel(p,q)$
2-connections for which $e \maps TM \to \fake$ is an isomorphism, and
hence a coframe field.  Pulling back along $e$, we thus get a metric
on $TM$ as well as a flat connection, the Weitzenb\"ock connection.
Its torsion is
\[
   T(v,w) = e^{-1}(d_\fc (v,w)). 
\]

We can then compute the action for a $\Tel(p,q)$ 2-connection using
the same formula given before in Equation (\ref{tele-action}):
\[
 S = \int d^nx\; \det(e) \left(
     \frac14 T^\rho{}_{\mu\nu} T_\rho{}^{\mu\nu} 
     + \frac12 T^\rho{}_{\mu\nu} T^{\nu\mu}{}_\rho 
     - T_{\rho\mu}{}^\rho T^{\nu\mu}{}_\nu
   \right)
\]
One should keep in mind that $\fc$ in this action is not an arbitrary
connection---it is part of a $\Tel(p,q)$ 2-connection, and hence
constrained to be flat.  In particular, in calculating the field
equations, only variations $\delta \fc$ preserving this flatness
constraint are allowed; since $\delta R = d_\fc \delta \fc$, these
$\delta \fc$ are precisely the covariantly closed ones.  An
alternative approach would be to allow $\fc$ to be an arbitrary
connection `off shell' but impose flatness when the equations of
motion hold, for example by using a Lagrange multiplier term.  
From the perspective of $\Tel(p,q)$ 2-connections, however, 
this seems less natural.

The above action is manifestly invariant under gauge transformations
of $\ff$: the determinant of $e$, the Weitzenb\"ock connection and its
torsion, as well as the metric, used to raise and lower indices, are
all invariant under gauge transformations acting on $\fc$ and $e$.

For purposes of comparing to other literature on teleparallel gravity,
however, it is worth noting that the action is often written in a way
that hides this gauge invariance.  From the perspective of this paper,
the apparent lack of $\O(p,q)$-invariance in many references on
teleparallel gravity (including, for example, the reference \cite{AP}
from which the formula for the action is taken), results from fixing
the flat connection on $\fake$ once and for all, by picking a fixed
trivialization $\fake = M\times \R^{p,q}$ and using the standard
flat connection on the trival $\R^{p,q}$ bundle.  If one transforms
the coframe field $e$ by a gauge transformation $g\maps M \to
\O(p,q)$, without also transforming the flat connection, the torsion,
and hence the action, will of course change.

\subsection{Cartan 2-geometry}

We now have a gravity action depending on a $\Tel(p,q)$ 2-connection.
However, is it invariant under (strict) gauge transformations of the
principal 2-bundle $\Tel(\ff)$?  Such a transformation is the same as
a gauge transformation of its bundle of objects, the principal
$\IO(p,q)$ bundle $\I\ff$.  Locally, these act on $e$ and $\omega$ to
give
\[
\begin{array}{ccll}
\om &\mapsto& h \om h^{-1}  +  h\, dh^{-1}  &=: \om' \\  
e &\mapsto& he \; + \; d_{\om'}  v  \\ 
%d_\om e &\mapsto& h d_{\om} e
\end{array}
\]
where $h$ and $v$ are functions with values in $\O(p,q)$ and
$\R^{p,q}$, respectively.  In the case where $h$ is the identity, this
amounts to shifting $e$ by an arbitrary covariantly exact 1-form. 
The action is {\em not} invariant under all such transformations, 
but only under
those for which $v = 0$.  These transformations are precisely the
gauge transformations of the fake 2-frame 2-bundle $\tff$.  To
see this, recall that the teleparallel 2-bundle can be built from
the fake 2-frame 2-bundle via
\[  
\Tel(\ff) \cong \tff \times_{\Po(p,q)} \Tel(p,q). 
\] 
Thus, any gauge transformation of $\tff$ gives one of $\Tel(\ff)$.
These are precisely the transformations with $v = 0$, since the gauge
2-group of $\tff$ has only Lorentz transformations as objects, not
translations.

In fact, something very similar often happens in ordinary
gauge-theoretic descriptions of gravity and related theories.  A
simple example is the {\bf Palatini action} for general relativity in
$n=p+q$ dimensions.  This action, which depends on a connection $\om$
on a fake tangent bundle, with curvature $R$, together with a coframe
field $e$, can be written:
\[
\label{palatini}
   S_{\scriptscriptstyle \rm Pal}[e,\om] = \int \star \big( \,\color{gray}\underbrace{\color{black}\rule[-.2em]{0em}{0em}e\we\cdots \we e}_{n-2}\color{black}\we R\,\big)
\]
The star operator $\star$ on the exterior bundle $\Lambda\fake$ turns
the $\Lambda^{n}\fake$-valued $n$-form in parentheses into an ordinary
real-valued $n$-form.  Using Proposition~\ref{iopq-conn}, we can view
the Palatini action as a function of a connection on the {\em
extended} fake frame bundle $\I\ff = \ff \times_{\O(p,q)} \IO(p,q)$.
However, the action is still invariant only under gauge
transformations of the subbundle $\ff$.

The reason for this apparently broken gauge symmetry in Palatini
gravity is that the $\IO(p,q)$ connection $(\om,e)$ is really a
`Cartan connection'.  Just as Riemannian geometry is the study of
spaces that look `infinitesimally' like Euclidean space, Cartan
geometry \cite{Sharpe} is the study of spaces that look
infinitesimally like homogeneous spaces.  It is thus an extension of
Klein's Erlangen program \cite{Klein} for understanding geometry using
homogeneous spaces.  Thus, it is important to note that the Erlangen
progam already uses a kind of `broken symmetry' to describe geometry.

In Klein's theory two groups play vital roles: a Lie group $G$ of
symmetries of a homogeneous space $X$, and a closed subgroup $G'
\subseteq G$, the stabilizer of an arbitrarily specified point, which
allows us to identify $X$ with $G/G'$.  So, a {\bf Klein geometry} is
technically a pair $(G,G')$ consisting of a Lie group and a closed
subgroup, but we think of these as a tool for studying the geometry of
the homogeneous space $G/G'$.

A Cartan geometry is then a space that is infinitesimally `modeled on'
$G/G'$.  We will not need a precise definition of Cartan geometry
here; what is important for our purposes is that it involves a principal 
$G'$ bundle $P$ together with a connection on
the associated principal $G$ bundle $P\times_{G'} G$.  Moreover, gauge
transformations of $P$ give isomorphic Cartan geometries, while more
general gauge transformations of $P\times_{G'} G$ can severely deform
the geometry.

In short: in a physical theory based on Cartan geometry, we expect to
see a $G$ connection for some Lie group $G$, but gauge invariance only
under some closed subgroup $G'$.  This phenomenon is ubiquitous in
attempts to describe gravity as a gauge theory by combining the
connection and coframe field into a larger connection, occurring not
only in Palatini gravity but also in MacDowell--Mansouri gravity
\cite{Wise2} and related theories \cite{Wise1}.  But the intriguing
fact we wish to emphasize is that \emph{the same thing seems to be
occurring in the higher gauge theory formulation of teleparallel
gravity---but with $G$ replaced by a 2-group}.

The present paper is not the place for extensive study of `Cartan
2-geometry'.  However, we would like to consider what a
straightforward reading of this analogy seems to imply for our
teleparallel gravity action.  As we have seen, in teleparallel gravity
the $\Po(p,q)$ 2-connection can be combined with the coframe
field $e$ to give a $\Tel(p,q)$ 2-connection $((\om,e),d_\om e)$, 
but the theory remains
invariant only under $\Po(p,q)$ gauge transformations.  This suggests
interpreting the theory in terms `Cartan 2-geometry' modeled on a
`homogeneous 2-space' given as the quotient $\Tel(p,q)/\Po(p,q)$.  But
what does this 2-space look like?

First, having done everything `strictly', it is easy to define a {\em
strict} quotient of strict 2-groups:

\begin{defn}
Let $\G$, be strict Lie 2-group with strict Lie sub-2-group $\G'$
(i.e.\ the groups of objects and morphisms of $\G'$ are Lie subgroups
of those of $\G$, and the maps are all restrictions of the
corresponding maps in the definition of $\G$).  The \define{strict
quotient} $\G/\G'$ is the Lie groupoid with
\begin{itemize}
\item $\G_0/\G'_0$ as objects 
\item $\G_1/\G'_1$ as morphisms
\item source, target, composition and identity-assigning maps induced
from those in $\G$.
\end{itemize}
\end{defn}
It is straightforward to check that the maps in $\G$ induce the corresponding 
maps on a well-defined way on the quotient, and that the result is indeed a 
Lie groupoid.  

There is a natural left action of $\G$ on the Lie groupoid $\G/\G'$,
induced by the left action of $\G$ on its underlying Lie groupoid (see
Example~\ref{G-2-space.2}).  By analogy with the 1-group case, we may
also refer to the strict quotient $\G/\G'$ as a `homogeneous 2-space'
for the group $\G$.  More generally, a quick way to define a
\define{homogeneous $\G$ 2-space} is to say it is any strict
$\G$ 2-space isomorphic to one of the form $\G/\G'$, though it is not
hard to give a more intrinsic definition.  Continuing with the
analogy to the ordinary case, we may also refer to the pair $(\G,\G')$
as a (\define{strict}) \define{Klein 2-geometry}.

\begin{prop}
$\Tel(p,q)/\Po(p,q)$ is isomorphic as a $\Tel(p,q)$ 2-space to the
space $\R^{p,q}$.
\end{prop}
\begin{proof}
Taking the strict quotient $\Tel(p,q)/\Po(p,q)$, we obtain a Lie
groupoid with:
\begin{itemize}
\item $\IO(p,q)/\O(p,q) \cong \R^{p,q}$ as objects, 
\item $(\IO(p,q)\ltimes \R^{p,q})/\IO(p,q) \cong \R^{p,q}$ as morphisms,
\item source and target maps $\R^{p,q} \to \R^{p,q}$ both the identity.
\end{itemize}
This is nothing but the space $\R^{p,q}$ thought of as a Lie groupoid
with only identity morphisms.  $\Tel(p,q)$ acts via its group of
objects $\IO(p,q)$, via the usual action of $\IO(p,q)$ on $\R^{p,q}$.
\end{proof}

\subsection{Weakening} 

In this paper, we have dealt entirely with {\em strict} constructions.
This approach eases the transition to higher gauge theory, since it is
built directly on constructions familiar from ordinary gauge theory.
For example, as we have seen:
\begin{itemize}
\item a strict 2-group has an ordinary group of objects;
\item a strict principal 2-bundle has an ordinary principal bundle of objects;
\item a strict 2-connection consists of an ordinary connection together with
a 2-form with values in an associated vector bundle; 
\item strict gauge transformations are induced by ordinary gauge
transformations.
\end{itemize}

However, even when one tries to keep everything `strict', as we have
done here, `weak' ideas can sneak in unexpectedly.  An
example is in the higher gauge theory interpretation of four-dimensional
\define{BF theory}.  As mentioned in the introduction, this is a gauge
theory for some Lie group $G$, with action given by
\[
S(A,B) = \int \tr \big( B \we F \big).
\]
Here, $A$ is a connection on a principal $G$ bundle $P$, $F$ it its
curvature, $B$ is a $(P\times_G \g)$-valued 2-form, and $G$ is assumed
semisimple, so that the Killing form `$\tr$' is nondegenerate.  This
theory can be viewed as a higher gauge theory for the \define{tangent
2-group} $\T G$, whose crossed module is $(G, \g, \Ad, 0)$: the fields
$(A,B)$ are the appropriate sort of ingredients for a 2-connection on
the strict 2-bundle $\P \times_G \T G$, and the field equations imply
the fake flatness condition.

However, while BF theory viewed as a higher gauge theory can be built
entirely on the {\em strict} constructions defined in this paper, it
has an extra symmetry that does not come from strict gauge
transformations of the principal $\T G$ 2-bundle.  If we shift $B$ by a
covariantly exact 2-form:
\[       
\begin{array}{ccl}
%A &\mapsto&  A   \\[.5em]  
B &\mapsto& B \; + \; d_{A}a 
\end{array}
\]
the action changes only by a boundary term, thanks to the Bianchi
identity $d_A F= 0$.  Remarkably, this additional symmetry implies
$BF$ theory is invariant under {\em weak} gauge transformations.

Given this lesson from BF theory, it is interesting to ask how
teleparallel gravity behaves under weak gauge transformations when we
regard it as a $\Tel(p,q)$ higher gauge theory, even though we have
done everything `strictly' so far.  For this, we need to know a bit
about how weak gauge transformations act on 2-connections.

General (weak) gauge symmetries of principal 2-bundles have been
described by Bartels \cite{Bartels:2004}.  For our purposes, it
suffices to recall, in the case where our principal 2-bundle is
equipped with a 2-connection, how the 2-connection data transform
under gauge transformations \cite{BS,BreenMessing}.  If $\G$ is a
2-group with crossed module $(G,H,t,\xa)$, then gauge transformations
on a trivial principal $\G$ 2-bundle act on a 2-connection $(A,B)$ to
give a new 2-connection $(A',B')$ with:
\begin{equation}
\label{weak-gauge}
\begin{array}{ccl}
A' &=& g A g^{-1} \; + \; g\, dg^{-1} \; + \; \dt(a)  \\[.5em]  
B' &=& \alpha(g)(B) \; + \; d_{A'}a \; + \; 
a \wedge a
\end{array}
\end{equation}
where $g$ is a $G$-valued function and $a$ is an $\h$-valued 1-form.
In the second equation, the covariant differential is defined by
$d_{A'} a = da + \dalpha(A')\we a$.  Strict gauge transformations
correspond to the case where $a=0$; the fully general ones are
called \define{weak}.

\begin{thm}
\label{local-gauge-equivalence}
If $\pi_1(M) = 0$, then all 2-connections on the trivial $\Tel(p,q)$
2-bundle over $M$ are equivalent under weak $\Tel(p,q)$ gauge
transformations.
\end{thm}
\begin{proof}
On a trivial $\Tel(p,q)$ 2-bundle, a 2-connection $(A,B)=((\fc,e),d_\om e)$
consists of an $\o(p,q)$-valued 1-form $\fc$, an $\R^{p,q}$-valued
1-form $e$, and the 2-form $B=d_\fc e$.  Specializing the formula
(\ref{weak-gauge}) for weak gauge transformations, we find that these
data transform to give a new 2-connection $(A',B')=((\om',e'),d_{\om'} e')$ 
with:
\begin{equation}
\label{weak-gauge.2}
\begin{array}{ccl}
\om' &=& h \om h^{-1}  +  h\, dh^{-1}   \\  
e' &=& he \; + \; d_{\om'}  v + a \\ 
B' &=& h d_{\om} e
\end{array}
\end{equation}
where $h\maps M\to \O(p,q)$, $v\maps M\to \R^{p,q}$ are smooth maps,
and $a\maps TM \to \R^{p,q}$ is an $\R^{p,q}$-valued 1-form on $M$.
In obtaining these equations we have used that $a \we a$ vanishes in
this case: it is defined using the Lie bracket, which vanishes on
$\R^{p,q}\subseteq \io(p,q)$.

Given a pair of 2-connections $(A,B)$ and $(A',B')$ on the trivial
$\Tel(p,q)$ 2-bundle over $M$, we wish to solve the above equations
for $h$, $v$, and $a$.  It is worth noting first that $v$ and $a$ do
not act in independent ways: we may clearly absorb $d_{\om'} v$ into
the definition of $a$, and hence without loss of generality assume
$v=0$.

Now $\om$ and $\om'$ are {\em flat} connections on the trivial
$\O(p,q)$ bundle over $M$.  But the moduli space of flat connections
on any fixed $\O(p,q)$ bundle is contained in
$\hom(\pi_1(M),\O(p,q))$, so $\pi_1(M) = 0$ implies $\om$ and $\om'$
must be related by a gauge transformation $h\maps M\to \O(p,q)$; that
is, $h\om h^{-1} + hdh^{-1} = \om'$.  This same $h$ changes $e$ to
$he$.  But, choosing $a=e'-he$, gives a gauge transformation mapping
$\om \mapsto \om'$ and $e\mapsto e'$.  An automatic consequence is
that $B\mapsto B'$.
\end{proof}

This theorem actually gives a `physics proof' that teleparallel
gravity, viewed as a $\Tel(p,q)$ higher gauge theory, cannot be
invariant under arbitrary weak gauge transformations.  If it were,
the theorem would imply the action was locally {\em independent} of the 
fields.  This cannot be true since teleparallel gravity has local degrees of 
freedom: indeed, it is locally 
equivalent to general relativity.  Of course, one can also check more
directly that the action is not invariant under the weak gauge
transformations given by Equation (\ref{weak-gauge.2}).

We found before that teleparallel gravity is invariant under strict
$\Po(p,q)$ gauge transformations.  The obvious question now is whether
it is also invariant under {\em weak} $\Po(p,q)$ gauge
transformations.  The answer to this question is already implicit in
the proof of Theorem~\ref{local-gauge-equivalence}, which really shows
a bit more than what the theorem states:

\begin{thm}
If $\pi_1(M) = 0$, then all 2-connections on the trivial $\Tel(p,q)$
2-bundle over $M$ are equivalent under weak $\Po(p,q)$ gauge
transformations.
\end{thm}
\begin{proof}
In the proof of Theorem~\ref{local-gauge-equivalence}, we noted that
in the transformations (\ref{weak-gauge.2}), we could take $v=0$
without loss of generality, and hence all $\Tel(p,q)$ 2-connections
are related by transformations of the form
\[
\begin{array}{ccl}
\om' &=& h \om h^{-1}  +  h\, dh^{-1}   \\  
e' &=& he \; + a \\ 
B' &=& h d_{\om} e
\end{array}
\]
These are just gauge transformations coming from the Poincar\'e
2-group $\Po(p,q)$.
\end{proof}

Summarizing our observations so far, consider this sequence of 2-groups:
\[    
    \O(p,q)  \longrightarrow \Po(p,q)  \longrightarrow \Tel(p,q)  
\]
The 2-connection of teleparallel gravity is a $\Tel(p,q)$
2-connection.  Considering strict gauge transformations, the action is
invariant only under $\Po(p,q)$ transformations, and this led us to
the suggestion that teleparallel gravity should be about `Cartan
2-geometry' based on the homogeneous 2-space $\Tel(p,q)/\Po(p,q)\cong
\R^{p,q}$.  On the other hand, considering weak gauge transformations,
the action is invariant only under $\O(p,q)$ transformations.  This
suggests `Cartan 2-geometry' based instead on the homogeneous 2-space
$\Tel(p,q)/\O(p,q)$.  We shall now see that this 2-space is the same
as the fundamental groupoid of $\R^{p,q}$.

\begin{defn}
Let $X$ be a manifold.  The \define{fundamental groupoid} of $X$, denoted
$\Pi_1(X)$, is the Lie groupoid whose objects are points of $X$ and whose
morphisms are homotopy classes of paths in $X$.
\end{defn}

For $\R^{p,q}$, any two points are connected by a unique homotopy
class of paths, so $\Pi_1(\R^{p,q})$ is particularly simple.  It is
clearly isomorphic to the Lie groupoid for which:
\begin{itemize}
\item $\R^{p,q}$ is the group of objects, 
\item $\R^{p,q}\times \R^{p,q}$ is the group of morphisms,
\item the source of the morphism $(v,w)$ is $v$,
\item the target of the morphism $(v,w)$ is $(v+w)$
\item the composite $(v,w) \circ (v',w')$, when defined, 
is $(v',w+w')$
\end{itemize}
Identifying $\Pi_1(\R^{p,q})$ with this Lie groupoid, we can turn
$\Pi_1(\R^{p,q})$ into a $\Tel(p,q)$ 2-space, defining an action
\[
    \begin{array}{ccc}
      \Tel(p,q) \times \Pi_1(\R^{p,q}) &\to& \Pi_1(\R^{p,q}) \\
    \end{array}
\]
given on objects by
\[
    \begin{array}{ccc}
      \IO(p,q) \times \R^{p,q} &\to& \R^{p,q} \\
       ((h,v),w) &\mapsto & hw+v
    \end{array}
\]
and given on morphisms as follows:
\[
    \begin{array}{ccc}
      (\IO(p,q)\ltimes\R^{p,q}) \times (\R^{p,q} \times \R^{p,q}) &\to& \R^{p,q}\times \R^{p,q} \\
       \big(((h,v),w), (v',w')\big) &\mapsto & (hv' + v, hw' + w) .
    \end{array}
\]
We then have:
\begin{thm}
As $\Tel(p,q)$ 2-spaces, $\Tel(p,q)/\O(p,q) \cong \Pi_1(\R^{p,q})$. 
\end{thm}
\begin{proof}
To form the strict quotient, 
we simply take the quotient of groups at both the object and 
the morphism level, obtaining a Lie groupoid with:
\begin{itemize}
\item $\IO(p,q)/\O(p,q) \cong \R^{p,q}$ as objects, and
\item $(\IO(p,q)\ltimes \R^{p,q})/\O(p,q) \cong \R^{p,q}\times \R^{p,q}$ as morphisms.
\end{itemize}
Comparing to the description of $\Tel(p,q)$ in Definition
\ref{defn:Tele}, it is clear that the above description of
$\Pi_1(\R^{p,q})$ results from ignoring the $\O(p,q)$ parts at
both object and morphism levels, and that the action of $\Tel(p,q)$ on
$\Pi_1(\R^{p,q})$ just comes from the left action of $\Tel(p,q)$ on
itself.
\end{proof}

\section{Outlook}
\label{conclusions}

We have seen that the geometry of teleparallel gravity is closely
related to the Poincar\'e 2-group, and to the teleparallel 2-group.
Previously, the Poincar\'e 2-group $\Po(p,q)$ seemed mathematically
natural but without much physical justification: why should we treat
Lorentz transformations as \textit{objects} but translations as
\textit{morphisms}?  Our answer is that if we do this, the
Weitzenb\"ock connection and its torsion fit together into a flat
$\Po(p,q)$ 2-connection.  Moreover, by extending $\Po(p,q)$ to the
teleparallel 2-group $\Tel(p,q)$, we can also include the coframe
field as part of the 2-connection, allowing us to write an action for
teleparallel gravity as a function of just a $\Tel(p,q)$ 2-connection.

On the other hand, we have seen that this action is invariant only
under gauge transformations living in a sub-2-group: $\Po(p,q)$ if we
consider only strict gauge transformations, or $\O(p,q)$ if we
consider weak ones.  We have discussed how this parallels the
geometric situation in many other approaches to gravity
\cite{Wise1,Wise2} where Cartan geometry \cite{Sharpe} provides the
geometric foundation, and suggested that `Cartan 2-geometry' may play
an analogous role in teleparallel gravity and other higher gauge
theories.

Indeed, we expect `Cartan 2-geometry' should be an interesting 
subject in its own right, with much broader 
applications than those suggested here.   In Cartan geometry, to each
homogeneous space $G/G'$, there corresponds a type of geometry 
that can be put on a more general manifold. 
Similarly, in Cartan 2-geometry each homogeneous 2-space $\G/\G'$
should give a type of geometry---or rather, a type of `2-geometry'---that 
can be put on a more general Lie groupoid.  

Our work suggests several things to be done toward 
developing a general theory of Cartan 2-geometry.  We would like to 
touch on two key points for this effort. 

First, while we have focussed on 2-bundles for which the `base 
2-space' is actually just a manifold, Cartan 2-geometry should in 
general involve 2-connections on bundles over intersting Lie groupoids, 
so many of the ideas we have discussed here deserve to be generalized
to that case.  

Second, despite our pragmatic use of {\em strict} constructions
throughout most of this paper, we ultimately expect the weak analogs
of concepts we have described to play a more fundamental role; weak
constructions in category theory are the most natural, and often the
most interesting.  In the teleparallel gravity case, we have seen that
we naturally get an interesting Lie groupoid as our `model 2-space'
only if we consider weak gauge transformations.  But `weakening' one
aspect of a theory tends to suggest, if not demand, weakening other
aspects.  For example, once we allow weak gauge transformations, we
can use them to assemble `weak 2-bundles' where the transition
functions on overlaps satisfy the usual equations only up to an
isomorphism.  In fact, even if we try building strict 2-bundles
initially, weak gauge transformations do not preserve strictness, so
we're essentially forced to use weak bundles.  Once we have weak
bundles, strict 2-connections no longer make sense.  And so on.

The theory of 2-bundles \cite{Bartels:2004} has been developed in
considerable generality, including weak principal bundles for {\em
weak} 2-groups, as well as the possibility of very general `base
2-spaces'---that is, base spaces that are Lie groupoids rather than
mere manifolds.  On the other hand, the theory of 2-connections, while
understood in a weak context \cite{BH, BS} is so far best understood
in the case where the base space is just a manifold.  Thus, the theory of
2-connections on 2-bundles over 2-spaces deserves further study.
Understanding `Cartan 2-connections' may be aided by studying the
concrete examples presented in this paper.

In the meantime, the heuristic picture is perhaps clear enough to 
venture a guess on some details of `Cartan 2-geometry', in the case
arising from teleparallel gravity, where the model `Klein 2-geometry' 
is
\[
   \Tel(p,q)/\O(p,q) \cong \Pi_1(\R^{p,q}).
\]
For this, we would like a 2-connection on a principal $\Tel(p,q)$
2-bundle that reduces to a principal $\O(p,q)$ 2-bundle.  It also
seems reasonable to take $\Pi_1(M)$ as the base 2-space, where $M$ is
a $(p+q)$-dimensional manifold.  In fact, there is an easy way to get
an $\O(p,q)$ 2-bundle over this Lie groupoid, and extend it to a
$\Tel(p,q)$ 2-bundle.  Start with a fake frame bundle $\ff\to M$,
equipped with a flat connection $\om$.  Form the 2-space $\Pi_1^{\rm
hor}(\ff, \om)$ whose objects are points in $\ff$ and whose morphisms
are homotopy classes of {\em horizontal} paths in $\ff$.  There is an
obvious projection
\[
\Pi_1^{\rm hor}(\ff, \om)\to \Pi_1(M),
\]
and it is easy to see that the object and morphism maps are both
principal $\O(p,q)$ bundles; in fact, this is a principal $\O(p,q)$
2-bundle.  Generalizing our associated 2-bundle construction to allow
for a base 2-space, the extension $\O(p,q) \to \Tel(p,q)$ gives a
$\Tel(p,q)$ 2-bundle.  Carrying on with the analogy to ordinary Cartan
geometry, we expect our Cartan 2-geometry to involve a 2-connection on
this 2-bundle, subject to certain `nondegeneracy' conditions.  However,
we leave the details for further work.

\subsection*{Acknowledgements}

We thank Aristide Baratin, John Huerta, and Jeffrey Morton for
helpful conversations regarding the Poincar\'e 2-group.  We are also 
grateful for the hospitality of several fine caf\'es in Erlangen 
where this work began.

\end{document}